\documentclass[12pt,english,onecolumn]{IEEEtran}
\usepackage[T1]{fontenc}
\usepackage{verbatim}
\usepackage{amsmath}
\usepackage{graphicx}
\usepackage{amssymb}

\makeatletter
\newtheorem{lemma}{Lemma}
\newtheorem{thm}{Theorem}
\newtheorem{remrk}{Remark}
\newtheorem{prop}{Proposition}

\usepackage{cite}
\usepackage{bm}
\usepackage{bbm}

\usepackage{fancyhdr}
\usepackage{babel}
\makeatother

\begin{document}

\title{Joint Beamforming for Multiaccess MIMO Systems with Finite Rate Feedback$^{*}$}

\author{Wei Dai$^{\dagger}$, Brian C. Rider$^{\dagger\dagger}$ and Youjian(Eugene)
Liu$^{\dagger}$\\
$^{\dagger}$Department of Electrical and Computer Engineering, University
of Colorado at Boulder\\
$^{\dagger\dagger}$Department of Mathematics, University of Colorado
at Boulder}

\maketitle
\begin{abstract}
This paper considers multiaccess multiple-input multiple-output (MIMO)
systems with finite rate feedback. The goal is to understand how to
efficiently employ the given finite feedback resource to maximize
the sum rate by characterizing the performance analytically. Towards
this, we propose a joint quantization and feedback strategy: the base
station selects the strongest users, jointly quantizes their strongest
eigen-channel vectors and broadcasts a common feedback to all the
users. This joint strategy is different from an individual strategy,
in which quantization and feedback are performed across users independently,
and it improves upon the individual strategy in the same way that
vector quantization improves upon scalar quantization. In our proposed
strategy, the effect of user selection is analyzed by extreme order
statistics, while the effect of joint quantization is quantified by
what we term {}``the composite Grassmann manifold''. The achievable
sum rate is then estimated by random matrix theory. Due to its simple
implementation and solid performance analysis, the proposed scheme
provides a benchmark for multiaccess MIMO systems with finite rate
feedback. 
\end{abstract}
\renewcommand{\thefootnote}{\fnsymbol{footnote}} \footnotetext[1]{This work is supported by NSF Grants CCF-0728955, ECCS-0725915, DMS-0505680 and  Thomson Inc. Part of content was presented in Allerton Conf. on Communication, Control, and Computing, 2005.} \renewcommand{\thefootnote}{\arabic{footnote}} \setcounter{footnote}{0}

\section{\label{sec:Introduction}Introduction}

This paper considers multiaccess systems, corresponding to the uplink
of cellular systems, where both the base station and the multiple
users are equipped with multiple antennas. Multiple antenna systems,
also known as multiple-input multiple-output (MIMO) systems, provide
significant benefit over single antenna systems in terms of increased
spectral efficiency and/or reliability. The full potential of MIMO
though requires perfect channel state information (CSI) at both the
transmitter and the receiver. While it is often reasonable to assume
that the receiver has perfect CSI through a pilot signal, assuming
perfect CSI at the transmitter (CSIT) is typically unrealistic. In
many practical systems, the transmitter obtains CSI through a finite
rate feedback from the receiver. Note that a wireless fading channel
may have infinitely many channel states, and a finite rate feedback
implies that CSIT is imperfect. One expects a performance degradation,
and here we focus on the quantitative effect of finite rate feedback
and the corresponding design.

Insight from single user MIMO systems with finite rate feedback proves
beneficial. Single user systems are similar to multiaccess systems
in the sense that there is only one receiver in both systems. The
receiver knows the channel states perfectly and helps transmitters
adapt their signals to maximize throughput. The essential difference
between these two types of systems lies in the modes of antenna cooperation.
In single user MIMO systems, all the transmit antennas are able to
cooperate in sending a given message. In multiaccess systems, different
users have independent messages, and transmit antennas belonging to
one user cannot aid the transmission of another user's message. Due
to this additional constraint, the analysis and design of multiaccess
systems becomes more complicated. Still, we will borrow insight from
single user systems to simplify the design of multiaccess systems.
For single user MIMO systems, strategies to maximize throughput with
perfect CSIT and without CSIT are derived and analyzed in \cite{Telatar_EuroTele99_Capacity_MIMO}.
When only finite rate feedback is available, the focus has moved toward
the development of suboptimal strategies as a simplification. The
dominant approach is based on power on/off strategy, in which a data
stream is either turned on with a pre-determined constant power or
turned off (zero power). Systems with only one stream are considered
in \cite{Sabharwal_IT03_Beamforming_MIMO,Love_IT03_Grassman_Beamforming_MIMO,Rao_SP07_Feedback_High_Resolution}.
Systems with multiple independent streams are investigated in \cite{Honig_Allerton03_Benefits_Limited_Feedback_Wireless_Channels,Rao_icc05_MIMO_spatial_multiplexing_limit_feedback,Love_IT2005_limited_feedback_unitary_precoding,Heath_ICASSP05_Quantization_Grassmann_Manifold,Dai_ISIT05_Power_onoff_strategy_design,Dai_05_Power_onoff_strategy_design_finite_rate_feedback,Dai_IT2008_Quantization_Grassmannian_manifold}.
It appears that power on/off strategy is near optimal compared to
the optimal power water-filling allocation \cite{Dai_05_Power_onoff_strategy_design_finite_rate_feedback}.

We aim to understand how to efficiently employ the given finite feedback
resource to maximize the sum rate by characterizing performance analytically.
The full multiaccess MIMO problem still appears behind reach mathematically
and is left for the future. In this paper, we propose a \emph{suboptimal}
strategy by borrowing insight and methods from single user systems.
Specifically, the base station selects the strongest users, jointly
quantizes their strongest eigen-channel vectors and broadcasts a common
feedback to all the users. Instead of designing a specific quantization
code book, we show that the performance of a random code book is optimal
in probability. After receiving feedback information, a selected on-user
employ power on/off strategy and transmit along the beamforming vector
selected by the feedback. Here, joint quantization and feedback are
employed based on the plain fact that vector quantization is better
than scalar quantization \cite[Ch. 13]{Cover_Elements_Information_Theory}.
(The precise gain will be verified empirically.) It is also worth
noting that, as we shall discuss in Section \ref{sec:Suboptimal-Feedback-Strategies}
and \ref{sec:Simulations-and-Discussion}, antenna selection can be
viewed as a simplified version of the proposed scheme.

This approach differs from the ongoing research for broadcast channels
(BC) with finite rate feedback. While there is a well known duality
between broadcast and multiaccess systems \cite{Vishwanath_IT2003_Duality_BC_MAC},
this duality requires full CSI at both the transmitters and the receivers
and is not available when only partial CSIT is provided. When CSIT
is available only through finite rate feedback, broadcast systems
suffer from the so called interference domination phenomenon \cite{Sharif_IT05_MIMO_BC_Feedback,Jindal_IT06_BC_Feedback}.
The major effort in research is to limit the interference among users.
Sharif and Hassibi select the near orthogonal channels when the number
of users is sufficiently large \cite{Sharif_IT05_MIMO_BC_Feedback,Jindal_IT06_BC_Feedback}.
As the number of users is comparable to the number of antennas at
the base station, Jindal shows that the feedback rate should be proportional
to signal-to-noise ratio (SNR) if the number of users turned on is
fixed \cite{Jindal_IT06_BC_Feedback}, while we show that the number
of users should be adapt to the SNR if the feedback rate is given
\cite{Dai_CISS2007_broadcast_multiaccess_channels_single_antenna_users}.
However, the interference domination phenomenon does not appear in
multiaccess systems. Note that the search of near orthogonal channels
suffers from exponential increasing complexity. Neither the results
nor the methods for broadcast systems can be directly applied to the
problem discussed in this paper.

Though the strategy in this paper is relatively simple, the corresponding
performance analysis is nontrivial. Our main analytical result is
an upper bound on the sum rate, which to our knowledge is the best
to date. The effect of user/antenna selection is analyzed by extreme
order statistics, and the effect of eigen-channel vectors joint quantization
is quantified via \emph{the composite Grassmann manifold}. Interestingly,
the complicated effect of imperfect CSIT and feedback is eventually
described by a single constant, which we term \emph{the power efficiency
factor}. Successful evaluation of the power efficiency factor enables
us characterize the upper bound on the sum rate. The anticipated goodness
of the upper bound is supported by simulation of several systems with
a large range of SNRs.

The rest of this paper is organized as follows. The general model
for multiaccess systems with finite rate feedback is described in
Section \ref{sec:System-Model}. The mathematical results developed
for performance analysis are assembled in Section \ref{sec:Mathematical-Results}.
The antenna selection strategy is analyzed in Section \ref{sub:Antenna-Selection}.
Then a suboptimal strategy is proposed and analyzed in Section \ref{sub:General-Beamforming-Strategy}.
In Section \ref{sec:Simulations-and-Discussion}, simulation results
are presented and discussed. Finally, Section \ref{sec:Conclusion}
summarizes the paper.

\section{\label{sec:System-Model}System Model}

Assume that there are $L_{R}$ antennas at the base station and $N$
users communicating with the base station. Assume that the user $i$%
\footnote{When a user joins the multiaccess system, a unique index is assigned
and keeps constant. A user in a multiaccess system is aware of the
corresponding index.%
} has $L_{T,i}$ transmit antennas $1\leq i\leq N$. Throughout we
will set $L_{T,1}=\cdots=L_{T,N}=L_{T}$. The signal transmission
model is \[
\mathbf{Y}=\sum_{i=1}^{N}\mathbf{H}_{i}\mathbf{T}_{i}+\mathbf{W},\]
 where $\mathbf{Y}\in\mathbb{C}^{L_{R}\times1}$ is the received signal
at the base station, $\mathbf{H}_{i}\in\mathbb{C}^{L_{R}\times L_{T}}$
is the channel state matrix for user $i$, $\mathbf{T}_{i}\in\mathbb{C}^{L_{T}\times1}$
is the transmitted Gaussian signal vector for user $i$ and $\mathbf{W}\in\mathbb{C}^{L_{R}\times1}$
is the additive Gaussian noise vector with zero mean and covariance
matrix $\mathbf{I}_{L_{R}}$. We assume the Rayleigh fading channel
model: the entries of $\mathbf{H}_{i}$'s are independent and identically
distributed (i.i.d.) circularly symmetric complex Gaussian variables
with zero mean and unit variance ($\mathcal{CN}\left(0,1\right)$),
and $\mathbf{H}_{i}$'s are independent across $i$.

We further assume that there exists a feedback link from the base
station to the users. At the beginning of each channel use, the channel
states $\mathbf{H}_{i}$'s are perfectly estimated at the receiver
(the base station). This assumption is valid in practice since most
communication standards allow the receiver to learn the channel states
from pilot signals. A common message, which is a function of the channel
states, is sent back to all users through the feedback link. We assume
that the feedback link is rate limited and error-free. The feedback
directs the users to choose their Gaussian signal covariance matrices.
In a multiaccess communication system, different users cannot cooperate
in terms of information message, leading to $\mathrm{E}\left[\mathbf{T}_{i}\mathbf{T}_{j}^{\dagger}\right]=\mathbf{0}$
for $i\ne j$. Let $\mathbf{T}=\left[\mathbf{T}_{1}^{\dagger}\cdots\mathbf{T}_{N}^{\dagger}\right]^{\dagger}$
be the overall transmitted Gaussian signal for all users and $\mathbf{\Sigma}\triangleq\mathrm{E}\left[\mathbf{T}\mathbf{T}^{\dagger}\right]$
be the overall signal covariance matrix. Then $\mathbf{\Sigma}$ is
an $NL_{T}\times NL_{T}$ block diagonal matrix whose $i^{\mathrm{th}}$
diagonal block is the $L_{T}\times L_{T}$ covariance matrix $\mathrm{E}\left[\mathbf{T}_{i}\mathbf{T}_{i}^{\dagger}\right]$.
Let $\mathbf{H}=\left[\mathbf{H}_{1}\mathbf{H}_{2}\cdots\mathbf{H}_{N}\right]$
be the overall channel state matrix. An extension of \cite{Lau_IT04_Capacity_Memoryless_Block_Fading}
shows that the optimal feedback strategy is to feedback the index
of an appropriate covariance matrix, which is a function of current
channel state $\mathbf{H}$. Last, assume that there is a covariance
matrix codebook $\mathcal{B}_{\mathbf{\Sigma}}=\left\{ \mathbf{\Sigma}_{1},\cdots,\mathbf{\Sigma}_{K_{\mathcal{B}}}\right\} $
(with finite size) declared to both the base station and the users,
where each $\mathbf{\Sigma}_{k}\in\mathcal{B}_{\mathbf{\Sigma}}$
is the overall signal covariance matrix with block diagonal structure
just described, and $K_{\mathcal{B}}$ is the size of the codebook.
The feedback function $\varphi$ is a map from $\left\{ \mathbf{H}\in\mathbb{C}^{L_{R}\times NL_{T}}\right\} $
onto the index set $\left\{ 1,\cdots,K_{\mathcal{B}}\right\} $. Subjected
to this finite rate feedback constraint\[
\left|\mathcal{B}_{\mathbf{\Sigma}}\right|=K_{\mathcal{B}}\]
 and the average total transmission power constraint\[
\mathrm{E}_{\mathbf{H}}\left[\mathrm{tr}\left(\mathbf{\Sigma}_{\varphi\left(\mathbf{H}\right)}\right)\right]\leq\rho,\]
 the sum rate of the optimal feedback strategy is given by \begin{equation}
\underset{\mathcal{B}_{\mathbf{\Sigma}}}{\sup}\;\underset{\varphi\left(\cdot\right)}{\sup}\;\mathrm{E}_{\mathbf{H}}\left[\log\left|\mathbf{I}_{L_{R}}+\mathbf{H}\mathbf{\Sigma}_{\varphi\left(\mathbf{H}\right)}\mathbf{H}^{\dagger}\right|\right].\label{eq:sum_rate_optimal}\end{equation}
 Here, since only symmetric systems are concerned, the total power
constraint $\rho$ is equivalent to individual power constraint $\rho/N$.
Note that the optimal strategy involves two coupled optimization problems.
It is difficult, if not impossible, to find its explicit form and
performance. Instead, we shall study two suboptimal strategies and
characterize their sum rates in Section \ref{sec:Suboptimal-Feedback-Strategies}.

\section{\label{sec:Mathematical-Results}Preliminaries}

This section assembles mathematical results required for later analysis.
The reader may proceed directly to Section \ref{sec:Suboptimal-Feedback-Strategies}
for the main engineering results.

\subsection{\label{sub:Extreme-Chi2}Order Statistics for Chi-Square Random Variables}

Define $X_{i}=\sum_{j=1}^{L}\left|h_{i,j}\right|^{2}$ where $h_{i,j}\;1\leq j\leq L,\;1\leq i\leq n$
are i.i.d. $\mathcal{CN}\left(0,1\right)$. Then each $X_{i}$ has
a Chi-square distribution with probability density functions (PDF)
\[
f_{X}\left(x\right)=\frac{1}{\left(L-1\right)!}x^{L-1}e^{-x}.\]
 Denote the corresponding cumulative distribution function (CDF) by
$F_{X}\left(x\right)$. Next introduce the order statistics for these
variables: that is the non-decreasing list $X_{\left(1:n\right)}\leq X_{\left(2:n\right)}\leq\cdots\leq X_{\left(n:n\right)}$
connected with each realization. Here, the subscript $\left(k:n\right)$
indicates that $X_{\left(k:n\right)}$ is the $k^{\mathrm{th}}$ minima.
(We follow the convention of \cite{JanosGalambos1987_extreme_order_statistics}.)
Note of course that ties occur with probability zero and can be broken
arbitrarily.

We will need the following, which is proved in Appendix \ref{sub:Proof-of-Extreme-Order-Statistics}.

\begin{lemma}
\label{lem:Expectation-extreme-chi2}With the notation set out above,
for any fixed positive integer $s$ it holds \begin{equation}
\underset{n\rightarrow+\infty}{\lim}\;\mathrm{E}\left[\frac{\sum_{k=1}^{s}X_{\left(n-k+1:n\right)}-sa_{n}}{b_{n}}\right]=s\left(\mu_{1}+1-\sum_{k=1}^{s}\frac{1}{k}\right),\label{eq:Expectation-extreme-chi2}\end{equation}
 where \[
a_{n}=\inf\left\{ x:\;1-F_{X}\left(x\right)\leq\frac{1}{n}\right\} ,\]
 \[
b_{n}=\frac{\sum_{i=0}^{L-1}\frac{L-i}{i!}a_{n}^{i}}{\sum_{i=0}^{L-1}\frac{1}{i!}a_{n}^{i}},\]
 and $\mu_{1}=\int_{-\infty}^{+\infty}xde^{-e^{-x}}$ may be computed
numerically. 

\end{lemma}

The limiting result in expectation immediately provides the following
approximation for a fixed $s$: \begin{equation}
\mathrm{E}\left[\sum_{k=1}^{s}X_{\left(n-k+1:n\right)}-sa_{n}\right]=sb_{n}\left(\mu_{1}+1-\sum_{i=1}^{s}\frac{1}{i}\right)\left(1+o\left(1\right)\right).\label{eq:approx-order-statistics}\end{equation}
 The shape of $F_{X}$ guarantees that $a_{n}$ and so $b_{n}$ are
finite for any fixed $n$ but tend to infinity and one respectively
with this parameter.

\subsection{\label{sub:Conditioned-Eigen}Conditioned Eigenvalues of the Wishart
Matrix}

Let $\mathbf{H}\in\mathbb{L}^{n\times m}$ be a random $n\times m$
matrix whose entries are i.i.d. Gaussian random variables with zero
mean and unit variance, where $\mathbb{L}$ is either $\mathbb{R}$
or $\mathbb{C}$. Throughout, we refer to $\mathbf{H}$ as the standard
Gaussian random matrix. Let $\lambda_{1}\ge\lambda_{2}\ge\cdots\ge\lambda_{n}$
be the ordered eigenvalues of $\mathbf{W}=\mathbf{H}\mathbf{H}^{\dagger}$
($\mathbf{W}$ is Wishart distributed \cite{Muirhead_book82_multivariate_statistics}).

This subsection takes up an estimate of $\mathrm{E}\left[\left.\lambda_{1}\right|\mathrm{tr}\left(\mathbf{W}\right)\right]$,
where $\mathrm{tr}\left(\cdot\right)$ is the usual matrix trace.
In particular, while a closed formula for this object would be rather
involved, we may use random matrix theory to obtain an approximation.
The first ingredient is the following.

\begin{lemma}
\label{lem:conditional-expectation-Wishart}Let $\mathbf{H}\in\mathbb{L}^{n\times m}$
(w.l.o.g. $n\le m$)%
\footnote{If $n>m$, $\mathrm{E}\left[\lambda_{i}|\mathrm{tr}\left(\mathbf{HH}^{\dagger}\right)=c\right]=0$
for $i>m$ and $\mathrm{E}\left[\lambda_{i}|\mathrm{tr}\left(\mathbf{HH}^{\dagger}\right)=c\right]=\zeta_{i}^{\prime}c$
for $i\le m$, where $\zeta_{i}^{\prime}:=\frac{1}{c}\mathrm{E}\left[\lambda_{i}|\mathrm{tr}\left(\mathbf{H}^{\dagger}\mathbf{H}\right)=c\right]$.
The calculation of $\zeta_{i}^{\prime}$ for $i\le m$ is included
in Theorem \ref{lem:conditional-expectation-Wishart} as well.%
} be a standard random Gaussian matrix. Let $\lambda_{1}\ge\lambda_{2}\ge\cdots\ge\lambda_{n}$
be the ordered eigenvalues of $\mathbf{W}=\mathbf{H}\mathbf{H}^{\dagger}$.
Then\[
\mathrm{E}\left[\lambda_{i}|\mathrm{tr}\left(\mathbf{W}\right)=c\right]=\zeta_{i}c,\]
 where\begin{equation}
\zeta_{i}=\mathrm{E}\left[\lambda_{i}|\mathrm{tr}\left(\mathbf{W}\right)=1\right].\label{eq:zeta-definition}\end{equation}
 \begin{equation}
,\label{eq:zeta-definition}\end{equation}
 $\beta=1$ if $\mathbb{L}=\mathbb{R}$ or $\beta=2$ if $\mathbb{L}=\mathbb{C}$,
and $\left|\Delta_{n}\left(\bm{\lambda}\right)\right|=\prod_{i<j}^{n}\left(\lambda_{i}-\lambda_{j}\right)$.
\end{lemma}
\begin{proof}
The joint density of the ordered eigenvalues of $\mathbf{W}$ is known
to be \[
K_{m,n,\beta}^{-1}e^{-\frac{\beta}{2}\sum_{i}\lambda_{i}}\prod_{i=1}^{m}\lambda_{i}^{\frac{\beta}{2}\left(n-m+1\right)-1}\left|\Delta_{n}\left(\bm{\lambda}\right)\right|^{\beta},\]
where $\lambda_{1}\geq\cdots\geq\lambda_{m}\geq0$, $\left|\Delta_{n}\left(\bm{\lambda}\right)\right|=\prod_{i<j}^{n}\left(\lambda_{i}-\lambda_{j}\right)$,
\[
\beta=\begin{cases}
1 & \mathrm{if}\;\mathbb{L}=\mathbb{R}\\
2 & \mathrm{if}\;\mathbb{L}=\mathbb{C}\end{cases},\]
and $K_{m,n,\beta}$ is a normalizing factor (\cite[pg. 107]{Muirhead_book82_multivariate_statistics}
for the real case and \cite{Telatar_EuroTele99_Capacity_MIMO} for
the complex case). Write out the formula for $\mathrm{E}\left[\lambda_{i}|\sum_{i=1}^{m}\lambda_{i}=c\right]$
and use the variable change $\lambda_{i}'=\frac{\lambda_{i}}{c}$.
After some elementary calculations, it can be shown that \[
\zeta_{i}=\frac{\int_{\underset{\lambda_{1}\geq\cdots\geq\lambda_{n}}{\sum\lambda_{j}=1}}\lambda_{i}\prod_{j=1}^{n}\lambda_{j}^{\frac{\beta}{2}\left(m-n+1\right)-1}\left|\Delta_{n}\left(\bm{\lambda}\right)\right|^{\beta}\prod_{j=1}^{n}d\lambda_{j}}{\int_{\underset{\lambda_{1}\geq\cdots\geq\lambda_{n}}{\sum\lambda_{j}=1}}\prod_{j=1}^{n}\lambda_{j}^{\frac{\beta}{2}\left(m-n+1\right)-1}\left|\Delta_{n}\left(\bm{\lambda}\right)\right|^{\beta}\prod_{j=1}^{n}d\lambda_{j}}=\mathrm{E}\left[\lambda_{i}|\mathrm{tr}\left(\mathbf{W}\right)=1\right].\]

\end{proof}

Given the preceding observation, we require an estimate for $\zeta_{1}$
in (\ref{eq:zeta-definition}). For this we turn to the asymptotic
behavior of the ordered eigenvalues.

\begin{lemma}
\label{lem:condition-expectation-asymptotics}Let $\lambda_{1}\ge\lambda_{2}\ge\cdots\ge\lambda_{n}$
be the ordered eigenvalues of $\frac{1}{m}\mathbf{H}\mathbf{H}^{\dagger}$,
where $\mathbf{H}\in\mathbb{L}^{n\times m}$ $\left(\mathbb{L}=\mathbb{R}\;\mathrm{or}\;\mathbb{C}\right)$
is a standard random Gaussian matrix. As $n,m\rightarrow\infty$ with
$\frac{m}{n}\rightarrow\bar{m}\in\mathbb{R}^{+}$, for a given $\tau\in\left(0,\min\left(1,\bar{m}\right)\right)$,
\begin{align*}
\bar{\zeta}_{\tau} & :=\underset{\left(n,m\right)\rightarrow\infty}{\lim}\mathrm{E}\left[\frac{1}{n}\left(\sum_{1\le i\le n\tau}\lambda_{i}\right)\right]\\
 & =\frac{\bar{m}}{2\pi}\left[\frac{1+\frac{1}{\bar{m}}-a}{2}\sqrt{\left(\lambda^{+}-a\right)\left(a-\lambda^{-}\right)}+\frac{2}{\bar{m}}\left(\frac{\pi}{2}+\sin^{-1}\frac{\sqrt{\bar{m}}\left(1+\frac{1}{\bar{m}}-a\right)}{2}\right)\right],\end{align*}
 where $\lambda^{+}=\left(1+\sqrt{\frac{1}{\bar{m}}}\right)^{2}$,
$\lambda^{-}=\left(1-\sqrt{\frac{1}{\bar{m}}}\right)^{2}$, and $a\in\left(\lambda^{-},\lambda^{+}\right)$
satisfies \begin{align*}
\tau & =\frac{\bar{m}}{2\pi}\left[-\sqrt{\left(\lambda^{+}-a\right)\left(a-\lambda^{-}\right)}+\frac{1+\bar{m}}{\bar{m}}\left(\frac{\pi}{2}+\sin^{-1}\frac{\sqrt{\bar{m}}\left(1+\frac{1}{\bar{m}}-a\right)}{2}\right)\right.\\
 & \quad\quad\quad\left.-\frac{\left|\bar{m}-1\right|}{\bar{m}}\left(\frac{\pi}{2}-\sin^{-1}\frac{\sqrt{\bar{m}}}{2}\frac{\left(1+\frac{1}{\bar{m}}\right)a-\left(1-\frac{1}{\bar{m}}\right)^{2}}{a}\right)\right].\end{align*}

\end{lemma}

This lemma is an extension of Theorem \ref{thm:truncate-function-RMT}
in Appendix \ref{sub:Random-Matrix-Theory} with explicit evaluation
of the integrals appearing in that statement. 

Motivated by the observation that the expectation of a fixed fraction
of the ordered eigenvalues converges to its limit quickly\cite{Dai_05_Power_onoff_strategy_design_finite_rate_feedback},
we approximate $\zeta_{1}$ by $\bar{\zeta}_{\frac{1}{n}}$ for fixed
finite $n$ and $m$.

\subsection{\label{sub:CGManifold}The Grassmann Manifold and the Composite Grassmann
Manifold}

As demonstrated in \cite{Dai_ISIT05_Power_onoff_strategy_design,Dai_05_Power_onoff_strategy_design_finite_rate_feedback},
the Grassmann manifold is closely related to eigen-channel vector
quantization, and here we introduce the composite Grassmann manifold.
The results developed here help quantify the effect of eigen-channel
vector quantization in multiaccess systems (see Section \ref{sub:General-Beamforming-Strategy}
for details).

The \emph{Grassmann manifold} $\mathcal{G}_{n,p}\left(\mathbb{L}\right)$
is the set of all $p$-dimensional planes (through the origin) in
the $n$-dimensional Euclidean space $\mathbb{L}^{n}$, where $\mathbb{L}$
is either $\mathbb{R}$ or $\mathbb{C}$. A generator matrix $\mathbf{P}\in\mathbb{L}^{n\times p}$
for a plane $P\in\mathcal{G}_{n,p}\left(\mathbb{L}\right)$ is a matrix
whose columns are orthonormal and span $P$. For a given $P\in\mathcal{G}_{n,p}\left(\mathbb{L}\right)$,
its generator matrix is not unique: if $\mathbf{P}$ generates $P$
then $\mathbf{PU}$ also generates $P$ for any $p\times p$ orthogonal/unitary
matrix $\mathbf{U}$ (with respect to $\mathbb{L}=\mathbb{R}/\mathbb{C}$
respectively) \cite{Conway_96_PackingLinesPlanes}. The chordal distance
between two planes $P_{1},P_{2}\in\mathcal{G}_{n,p}\left(\mathbb{L}\right)$
can be defined by their generator matrices $\mathbf{P}_{1}$ and $\mathbf{P}_{2}$
via \[
d_{c}\left(P_{1},P_{2}\right)=\sqrt{p-\mathrm{tr}\left(\mathbf{P}_{1}^{\dagger}\mathbf{P}_{2}\mathbf{P}_{2}^{\dagger}\mathbf{P}_{1}\right)}.\]
 The isotropic measure $\mu$ on $\mathcal{G}_{n,p}\left(\mathbb{L}\right)$
is the Haar measure on $\mathcal{G}_{n,p}\left(\mathbb{L}\right)$%
\footnote{The Haar measure is well defined for locally compact topological groups
\cite{Haar_1933_Haar_Measure,Muirhead_book82_multivariate_statistics},
and therefore for the Grassmann manifold, the composite Grassmann
manifold and the composite Grassmann matrices. Here, the group right
and left operations are clear from context. %
}. Let $O\left(n\right)$ and $U\left(n\right)$ be the sets of $n\times n$
orthogonal and unitary matrices respectively. Let $\mathbf{A}\in O\left(n\right)$
when $\mathbb{L}=\mathbb{R}$, or $\mathbf{A}\in U\left(n\right)$
when $\mathbb{L}=\mathbb{C}$. For any measurable set $\mathcal{M}\subset\mathcal{G}_{n,p}\left(\mathbb{L}\right)$
and arbitrary $\mathbf{A}$, $\mu$ satisfies \[
\mu\left(\mathbf{A}\mathcal{M}\right)=\mu\left(\mathcal{M}\right).\]

Given above definitions, the distortion rate tradeoff on the Grassmann
manifold is quantified in \cite{Dai_Globecom05_Quantization_bounds_Grassmann_manifold,Dai_IT2008_Quantization_Grassmannian_manifold}.
A quantization $\mathfrak{q}$ is a mapping from $\mathcal{G}_{n,p}\left(\mathbb{L}\right)$
to a discrete subset of $\mathcal{G}_{n,p}\left(\mathbb{L}\right)$,
which is typically called a code $\mathcal{C}$. For the sake of application,
the quantization \begin{align*}
\mathfrak{q}:\;\mathcal{G}_{n,p}\left(\mathbb{L}\right) & \rightarrow\mathcal{C}\\
Q & \mapsto\mathfrak{q}\left(Q\right)=\arg\;\underset{P\in\mathcal{C}}{\min}\; d_{c}\left(P,Q\right)\end{align*}
 is of particular interest. Define the distortion metric on $\mathcal{G}_{n,p}\left(\mathbb{L}\right)$
as the squared chordal distance. Let $Q\in\mathcal{G}_{n,p}\left(\mathbb{L}\right)$
be isotropically distributed (the probability measure is the isotropic
measure). For a given code $\mathcal{C}$, the distortion associated
with this codebook is defined as \[
D\left(\mathcal{C}\right)=\mathrm{E}_{Q}\left[\underset{P\in\mathcal{C}}{\min}\; d_{c}^{2}\left(P,Q\right)\right].\]
 For a given code size $K$ where $K$ is a positive integer, the
distortion rate function is \[
D^{*}\left(K\right)=\underset{\mathcal{C}:\left|\mathcal{C}\right|=K}{\inf}\; D\left(\mathcal{C}\right).\]
 In \cite{Dai_Globecom05_Quantization_bounds_Grassmann_manifold,Dai_IT2008_Quantization_Grassmannian_manifold},
we quantify the distortion rate function by constructing tight lower
and upper bounds. The results are summarized as follows.

\begin{lemma}
\label{lem:DRF-GM}Consider the distortion rate function on $\mathcal{G}_{n,p}\left(\mathbb{L}\right)$.
Let $t=\beta p\left(n-p\right)$,\[
\beta=\left\{ \begin{array}{ll}
1 & \mathrm{if}\;\mathbb{L}=\mathbb{R}\\
2 & \mathrm{if}\;\mathbb{L}=\mathbb{C}\end{array}\right.,\]
\[
c_{n,p,p,\beta}=\begin{cases}
\frac{1}{\Gamma\left(\frac{t}{2}+1\right)}\prod_{i=1}^{p}\frac{\Gamma\left(\frac{\beta}{2}\left(n-i+1\right)\right)}{\Gamma\left(\frac{\beta}{2}\left(p-i+1\right)\right)} & \mathrm{if}\; p\le\frac{n}{2}\\
\frac{1}{\Gamma\left(\frac{t}{2}+1\right)}\prod_{i=1}^{n-p}\frac{\Gamma\left(\frac{\beta}{2}\left(n-i+1\right)\right)}{\Gamma\left(\frac{\beta}{2}\left(n-p-i+1\right)\right)} & \mathrm{otherwise}\end{cases}.\]
When $K$ is sufficiently large ($c_{n,p,p,\beta}^{-\frac{2}{t}}2^{-\frac{2\log_{2}K}{t}}\le1$
necessarily), \begin{align}
 & \frac{t}{t+2}c_{n,p,p,\beta}^{-\frac{2}{t}}2^{-\frac{2\log_{2}K}{t}}\left(1+o\left(1\right)\right)\leq D^{*}\left(K\right)\nonumber \\
 & \quad\quad\quad\leq\frac{2}{t}\Gamma\left(\frac{2}{t}\right)c_{n,p,p,\beta}^{-\frac{2}{t}}2^{-\frac{2\log_{2}K}{t}}\left(1+o\left(1\right)\right).\label{eq:DRF_bounds_GM}\end{align}

\end{lemma}

To analyze the joint quantization problem arising in multiaccess MIMO
systems (see Section \ref{sub:General-Beamforming-Strategy} for details),
we introduce the \emph{composite Grassmann manifold}. The $m$-composite
Grassmann manifold $\mathcal{G}_{n,p}^{\left(m\right)}\left(\mathbb{L}\right)$
is a Cartesian product of $m$ many $\mathcal{G}_{n,p}\left(\mathbb{L}\right)$'s.
Denote $P^{\left(m\right)}$ an element in $\mathcal{G}_{n,p}^{\left(m\right)}\left(\mathbb{L}\right)$.
Then $P^{\left(m\right)}$ can be written as $\left(P_{1},\cdots,P_{m}\right)$
where $P_{i}\in\mathcal{G}_{n,p}\left(\mathbb{L}\right)$ $1\leq i\leq m$.
For any $P_{1}^{\left(m\right)},P_{2}^{\left(m\right)}\in\mathcal{G}_{n,p}^{\left(m\right)}\left(\mathbb{L}\right)$,
the chordal distance between them is well defined by \[
d_{c}\left(P_{1}^{\left(m\right)},P_{2}^{\left(m\right)}\right):=\sqrt{\sum_{j=1}^{m}d_{c}^{2}\left(P_{1,j},P_{2,j}\right)},\]
 where $P_{1}^{\left(m\right)}=\left(P_{1,1},\cdots,P_{1,m}\right)$,
$P_{2}^{\left(m\right)}=\left(P_{2,1},\cdots,P_{2,m}\right)$ and
$P_{i,j}\in\mathcal{G}_{n,p}\left(\mathbb{L}\right)$ ($i=1,2$ and
$j=1,2,\cdots,m$). The isotropic measure on $\mathcal{G}_{n,p}^{\left(m\right)}\left(\mathbb{L}\right)$
can be induced from the isotropic measure on $\mathcal{G}_{n,p}\left(\mathbb{L}\right)$:
it is just the product of the isotropic measures on the composed copies
of $\mathcal{G}_{n,p}\left(\mathbb{L}\right)$.

One goal will be to characterize the distortion rate function on $\mathcal{G}_{n,p}^{\left(m\right)}\left(\mathbb{L}\right)$.
By analogy with the above discussion let a code $\mathcal{C}$ be
any discrete subset of $\mathcal{G}_{n,p}^{\left(m\right)}\left(\mathbb{L}\right)$,
and consider the quantization function \begin{equation}
\mathfrak{q}\left(Q^{\left(m\right)}\right)=\underset{P^{\left(m\right)}\in\mathcal{C}}{\arg\;\min}\; d_{c}\left(P^{\left(m\right)},Q^{\left(m\right)}\right).\label{eq:quantization-fn-CGM}\end{equation}
 Let the distortion metric on $\mathcal{G}_{n,p}^{\left(m\right)}\left(\mathbb{L}\right)$
be the squared chordal distance. The distortion associated with $\mathcal{C}$
is given by \[
D\left(\mathcal{C}\right)=\mathrm{E}_{Q^{\left(m\right)}}\left[\underset{P^{\left(m\right)}\in\mathcal{C}}{\min}\; d_{c}^{2}\left(P^{\left(m\right)},Q^{\left(m\right)}\right)\right],\]
 where $Q^{\left(m\right)}\in\mathcal{G}_{n,p}^{\left(m\right)}\left(\mathbb{L}\right)$
is isotropically distributed. For all $K\in\mathbb{Z}^{+}$, the distortion
rate function is defined as \[
D^{*}\left(K\right)=\underset{\mathcal{C}:\left|\mathcal{C}\right|=K}{\inf}\; D\left(\mathcal{C}\right).\]
 The following theorem characterizes $D^{*}\left(K\right)$ on $\mathcal{G}_{n,p}^{\left(m\right)}\left(\mathbb{L}\right)$.

\begin{thm}
\label{thm:DRF-CGM}Consider the distortion rate function on $\mathcal{G}_{n,p}^{\left(m\right)}\left(\mathbb{L}\right)$.
Let $t$, $c_{n,p,p,\beta}$ and $\beta$ be defined as in Lemma \ref{lem:DRF-GM}.
When $K$ is sufficiently large ($\frac{\Gamma^{\frac{2}{mt}}\left(mt+1\right)}{\Gamma^{\frac{2}{t}}\left(t+1\right)}c_{n,p,p,\beta}^{-\frac{2}{t}}2^{-\frac{2\log_{2}K}{mt}}\le1$
necessarily), \begin{align}
 & \frac{mt}{mt+2}\left(\frac{\Gamma^{\frac{2}{mt}}\left(m\frac{t}{2}+1\right)}{\Gamma^{\frac{2}{t}}\left(\frac{t}{2}+1\right)}c_{n,p,p,\beta}^{-\frac{2}{t}}2^{-\frac{2\log_{2}K}{mt}}\right)\left(1+o\left(1\right)\right)\le D^{*}\left(K\right)\nonumber \\
 & \quad\quad\quad\le\frac{2}{mt}\Gamma\left(\frac{2}{mt}\right)\left(\frac{\Gamma^{\frac{2}{mt}}\left(m\frac{t}{2}+1\right)}{\Gamma^{\frac{2}{t}}\left(\frac{t}{2}+1\right)}c_{n,p,p,\beta}^{-\frac{2}{t}}2^{-\frac{2\log_{2}K}{mt}}\right)\left(1+o\left(1\right)\right).\label{eq:DRF-CGM}\end{align}

\end{thm}

The detailed proof is given in Appendix \ref{sub:proof_DRF_composite_GM},
but we mention here that the upper bound is derived by calculating
the average distortion of random codes, which turn out to be asymptotically
optimal in probability. Further, the lower and upper bounds differ
only in the constant factors: $\frac{mt}{mt+2}$ for the lower bound
and $\frac{2}{mt}\Gamma\left(\frac{2}{mt}\right)$ for the upper bound.
As $n,K\rightarrow\infty$ with $\frac{\log_{2}K}{n}\rightarrow r$,
this discrepency vanishes and we precisely characterize the asymptotic
distortion rate function.

\begin{thm}
\label{thm:DRF-CGM-Asymptotics}Fix $p$ and $m$. Let $n,K\rightarrow\infty$
with $\frac{\log_{2}K}{n}\rightarrow r$. If $r$ is sufficiently
large ($mp2^{-\frac{2}{\beta mp}r}\le1$ necessarily), then \[
\underset{\left(n,K\right)\rightarrow\infty}{\lim}D^{*}\left(K\right)=mp2^{-\frac{2}{\beta mp}r},\]
where $\beta=1$ if $\mathbb{L}=\mathbb{R}$, and $\beta=2$ if $\mathbb{L}=\mathbb{C}$.
Furthermore, let $\mathcal{C}_{\mathrm{rand}}\subset\mathcal{G}_{n,p}^{\left(m\right)}\left(\mathbb{L}\right)$
be a code random generated from the isotropic distribution and with
size $K$. Then for $\forall\epsilon>0$, \[
\underset{\left(n,K\right)\rightarrow\infty}{\lim}\Pr\left(D\left(\mathcal{C}_{\mathrm{rand}}\right)>mp2^{-\frac{2}{\beta mp}r}+\epsilon\right)=0.\]

\end{thm}
{}

The proof of this theorem follows from those in \cite[Theorem 3]{Dai_IT2008_Quantization_Grassmannian_manifold}
and is omitted here.

This theorem provides a formula for the distortion rate function at
finite $n$ and $K$: \begin{equation}
D^{*}\left(K\right)=\frac{2}{mt}\Gamma\left(\frac{2}{mt}\right)\left(\frac{\Gamma^{\frac{2}{mt}}\left(m\frac{t}{2}+1\right)}{\Gamma^{\frac{2}{t}}\left(\frac{t}{2}+1\right)}c_{n,p,p,\beta}^{-\frac{2}{t}}2^{-\frac{2\log_{2}K}{mt}}\right)\left(1+o\left(1\right)\right).\label{eq:DRF-CGM-approx}\end{equation}
 By the asymptotic optimality of random codes, we have employed random
codes for our analysis, and approximate the corresponding distortion
rate function by ignoring the higher order terms behind (\ref{eq:DRF-CGM-approx}).

\subsection{\label{sub:CGMatrix}Calculations Related to Composite Grassmann
Matrices}

A \emph{composite Grassmann matrix} $\mathbf{P}^{\left(m\right)}$
is a generator matrix generating $P^{\left(m\right)}\in\mathcal{G}_{n,p}^{\left(m\right)}\left(\mathbb{L}\right)$,
and we denote the set of composite Grassmann matrices by $\mathcal{M}_{n,p}^{\left(m\right)}\left(\mathbb{L}\right)$.
A composite Grassmann matrix $\mathbf{P}^{\left(m\right)}=\left[\mathbf{P}_{1}\cdots\mathbf{P}_{m}\right]\in\mathcal{M}_{n,p}^{\left(m\right)}\left(\mathbb{L}\right)$
generates a plane $P^{\left(m\right)}=\left(P_{1},\cdots,P_{m}\right)\in\mathcal{G}_{n,p}^{\left(m\right)}\left(\mathbb{L}\right)$,
where $\mathbf{P}_{1},\cdots,\mathbf{P}_{m}$ are the generator matrices
for $P_{1},\cdots,P_{m}$ respectively. Note that the generator matrix
$\mathbf{P}_{i}$ for a plane $P_{i}\in\mathcal{G}_{n,p}\left(\mathbb{L}\right)$
is not unique. The composite Grassmann matrix $\mathbf{P}^{\left(m\right)}\in\mathcal{M}_{n,p}^{\left(m\right)}\left(\mathbb{L}\right)$
generating $P^{\left(m\right)}\in\mathcal{G}_{n,p}^{\left(m\right)}\left(\mathbb{L}\right)$
is not unique either: let $\mathbf{U}^{\left(m\right)}$ is an arbitrary
$pm\times pm$ block diagonal matrix whose $i^{\mathrm{th}}$ ($1\le i\le m$)
diagonal block is a $p\times p$ orthogonal/unitary matrix (w.r.t.
$\mathbb{L}=\mathbb{R}/\mathbb{C}$ respectively); if $\mathbf{P}^{\left(m\right)}$
generates $P^{\left(m\right)}$, then $\mathbf{P}^{\left(m\right)}\mathbf{U}^{\left(m\right)}$
generates $P^{\left(m\right)}$ as well. View $\mathcal{M}_{n,p}^{\left(m\right)}\left(\mathbb{L}\right)$
as a Cartesian product of $m$ many $\mathcal{M}_{n,p}^{\left(1\right)}\left(\mathbb{L}\right)$.
Then the isotropic measure $\mu$ on $\mathcal{M}_{n,p}^{\left(m\right)}\left(\mathbb{L}\right)$
is simply the product of Haar measure on each composed $\mathcal{M}_{n,p}^{\left(1\right)}\left(\mathbb{L}\right)$'s.
We say a $\mathbf{P}^{\left(m\right)}\in\mathcal{M}_{n,p}^{\left(m\right)}\left(\mathbb{L}\right)$
is isotropically distributed if the corresponding probability measure
is the isotropic measure $\mu$.

Note now that we are interested in quantifying $\mathrm{E}\left[\log\left|\mathbf{I}+c\mathbf{P}^{\left(m\right)}\mathbf{P}^{\left(m\right)\dagger}\right|\right]$,
for a constant $c\in\mathbb{R}^{+}$ and isotropically distributed
$\mathbf{P}^{\left(m\right)}\in\mathcal{M}_{n,1}^{\left(m\right)}\left(\mathbb{C}\right)$.
The asymptotic behavior of this quantify is derived by random matrix
theory techniques.

\begin{thm}
\label{thm:bds_CGMatrix}Let $\mathbf{P}^{\left(m\right)}\in\mathcal{M}_{n,1}^{\left(m\right)}\left(\mathbb{C}\right)$
be isotropically distributed. For all positive real numbers $c$,
as $n,m\rightarrow\infty$ with $\frac{m}{n}\rightarrow\bar{m}\in\mathbb{R}^{+}$,\begin{align}
 & \underset{\left(n,m\right)\rightarrow\infty}{\lim}\frac{1}{n}\mathrm{E}\left[\log\left|\mathbf{I}+c\mathbf{P}^{\left(m\right)}\mathbf{P}^{\left(m\right)\dagger}\right|\right]\nonumber \\
 & =\log\left(1+c\bar{m}-\frac{1}{4}\mathcal{F}\left(c,\bar{m}\right)\right)+\bar{m}\log\left(1+c-\frac{1}{4}\mathcal{F}\left(c,\bar{m}\right)\right)-\frac{\mathcal{F}\left(c,\bar{m}\right)}{4c},\label{eq:CGM-Shannon-Transform}\end{align}
where \[
\mathcal{F}\left(z,\bar{m}\right)=\left(\left(1+\lambda^{-}z\right)^{1/2}-\left(1+\lambda^{+}z\right)^{1/2}\right)^{2},\]
$\lambda^{+}=\left(1+\sqrt{1/\bar{m}}\right)^{2}$ and $\lambda^{-}=\left(1-\sqrt{1/\bar{m}}\right)^{2}$. 
\end{thm}
\begin{proof}
Let $\mathbf{H}\in\mathbb{C}^{n\times m}$ be a standard Gaussian
matrix. Let $\mathbf{P}^{\left(m\right)}\in\mathcal{G}_{n,1}^{\left(m\right)}\left(\mathbb{C}\right)$
be isotropically distributed. As $n,m\rightarrow\infty$ with a positive
ratio, the eigenvalue statistics of $\mathbf{P}^{\left(m\right)}\mathbf{P}^{\left(m\right)\dagger}$
and $\frac{1}{m}\mathbf{H}\mathbf{H}^{\dagger}$ are asymptotically
the same. Indeed, the Raleigh-Ritz criteria shows that the discrepancy
between corresponding eigenvalues of these two matrices is bounded
(multiplicatively) above and below by the minimum and maximum column
norms of $\frac{1}{m}\mathbf{H}\mathbf{H}^{\dagger}$, both of which
converge to one almost surely. Thus,\[
\underset{\left(n,m\right)\rightarrow\infty}{\lim}\frac{1}{n}\mathrm{E}\left[\log\left|\mathbf{I}+c\mathbf{P}^{\left(m\right)}\mathbf{P}^{\left(m\right)\dagger}\right|\right]=\underset{\left(n,m\right)\rightarrow\infty}{\lim}\frac{1}{n}\mathrm{E}\left[\log\left|\mathbf{I}+c\frac{m}{n}\frac{1}{m}\mathbf{H}\mathbf{H}^{\dagger}\right|\right].\]
 Now, it is a basic result in random matrix theory \cite[Eq. (1.10)]{Verdu_random_matrix_theory_wireless_communications}
(also see Appendix \ref{sub:Random-Matrix-Theory}) that the empirical
distribution of the eigenvalues of $\frac{1}{m}\mathbf{H}\mathbf{H}^{\dagger}$
converges to the Mar\u{c}enko-Pastur law given by \[
d\mu_{\lambda}=\left(1-\bar{m}\right)^{\dagger}\delta\left(\lambda\right)+\bar{m}\frac{\sqrt{\left(\lambda-\lambda^{-}\right)^{+}\left(\lambda^{+}-\lambda\right)^{+}}}{2\pi\lambda}\cdot d\lambda\]
 almost surely, where $\left(z\right)^{+}=\max\left(0,z\right)$.
Thus, \[
\underset{\left(n,m\right)\rightarrow\infty}{\lim}\frac{1}{n}\mathrm{E}\left[\log\left|\mathbf{I}+c\mathbf{P}^{\left(m\right)}\mathbf{P}^{\left(m\right)\dagger}\right|\right]\rightarrow\int\log\left(1+c\bar{m}\lambda\right)\cdot d\mu_{\lambda}\]
 since $\log\left(1+c\bar{m}\lambda\right)$ is a bounded continuous
function on the spectral support. The resulting integral is evaluated
in \cite{Verdu_IT1999_Spectral_efficiency_CDMA}, and the proof is
finished.

%
{}
\end{proof}

For finite $n$ and $m$, we substitute $\bar{m}=\frac{m}{n}$ into
(\ref{eq:CGM-Shannon-Transform}) to approximate $\frac{1}{n}\mathrm{E}\left[\log\left|\mathbf{I}+c\mathbf{P}^{\left(m\right)}\mathbf{P}^{\left(m\right)\dagger}\right|\right]$.

\section{\label{sec:Suboptimal-Feedback-Strategies}Suboptimal Strategies
and the Sum Rate}

Given finite rate feedback, the optimal strategy (\ref{eq:sum_rate_optimal})
involves two coupled optimization problems: one is with respect to
the feedback function $\varphi$ and the other optimization is over
all possible covariance matrix codebooks. The corresponding design
and analysis are extremely complicated, and instead we study suboptimal
power on/off strategies. Motivated by the near optimal power on/off
strategy for single user MIMO systems \cite{Dai_ISIT05_Power_onoff_strategy_design,Dai_05_Power_onoff_strategy_design_finite_rate_feedback},
we assume:

\begin{description}
\item [{T1)}] Power on/off strategy: The $i$th user covariance matrix
is of the form $\mathbf{\Sigma}_{i}=P_{\mathrm{on}}\mathbf{Q}_{i}\mathbf{Q}_{i}^{\dagger}$,
where $P_{\mathrm{on}}$ is a fixed positive constant to denote on-power
and $\mathbf{Q}_{i}$ is the beamforming matrix for user $i$. Denote
each column of $\mathbf{Q}_{i}$ an \emph{on-beam} and the number
of the columns of $\mathbf{Q}_{i}$ by $s_{i}$ ($0\leq s_{i}\leq L_{T}$),
then $\mathbf{Q}_{i}^{\dagger}\mathbf{Q}_{i}=\mathbf{I}_{s_{i}}$.
Here, $s_{i}$ is the number of data streams (or on-beams) for user
$i$ ($s_{i}=0$ implies that the user $i$ is off). The user $i$
with $s_{i}>0$ is referred to as an on-user. 
\item [{T2)}] Constant number of on-beams: Let $s=\sum_{i=1}^{N}s_{i}$,
the total number of on-beams, be constant independent of the specific
channel realization for a given SNR. With this assumption, $P_{\mathrm{on}}=\rho/s$. 
\end{description}
\begin{remrk}
Using a constant number of on-beams is motivated by the fact that
it is asymptotically optimal to turn on a constant fraction of all
eigen-channels as $L_{T},L_{R}\rightarrow\infty$ with a positive
ratio, see \cite{Dai_05_Power_onoff_strategy_design_finite_rate_feedback}
which also demonstrates the good performance of this strategy. While
the number of on-beams is independent of channel realizations, it
is a function of SNR. Realize though that typically SNR changes on
a much larger time scale than block fading. Keeping the number of
on-beams constant enables the base station to keep the feedback and
decoding processing from one fading block to another, and therefore
reduces complexity of real-world systems.

%
{}
\end{remrk}

These two assumptions essentially add extra structure to the input
covariance matrix $\mathbf{\Sigma}$. Given this structure, we propose
a joint quantization and feedback strategy in Section \ref{sub:General-Beamforming-Strategy},
which we term {}``general beamforming strategy\char`\"{}. As we shall
see in Section \ref{sub:Comments}, antenna selection can be viewed
as a special case of general beamforming. Due to the simplicity of
antenna selection, we next discuss its main features.

\subsection{\label{sub:Antenna-Selection}Antenna Selection}

The antenna selection strategy is described as follows. Index all
$NL_{T}$ antennas by $i$ ($i=1,\cdots,NL_{T}$). Then \[
\mathbf{Y}=\sum_{i=1}^{NL_{T}}\mathbf{h}_{i}X_{i}+\mathbf{W},\]
where $\mathbf{h}_{i}$ is the $i^{\mathrm{th}}$ column of the overall
channel state matrix $\mathbf{H}$ (defined in Section \ref{sec:System-Model}),
and $X_{i}$ is the Gaussian data source corresponding to the antenna
$i$. Power on/off assumptions (T1) and (T2) imply that either $\mathrm{E}\left[X_{i}^{2}\right]=\frac{\rho}{s}$
or $\mathrm{E}\left[X_{i}^{2}\right]=0$. Indeed, for a specific user,
its input covariance matrix can be written as $\frac{\rho}{s}\mathbf{Q}\mathbf{Q}^{\dagger}$
where $\mathbf{Q}$ is obtained from intercepting some columns from
the identity matrix. Given a channel realization $\mathbf{H}$, the
base station selects $s$ many antennas according to 

\begin{description}
\item [{F1)}] Antenna selection criterion. Sort the channel state vectors
$\mathbf{h}_{i}$'s increasingly according to their Frobenius norms
such that $\left\Vert \mathbf{h}_{\left(1:NL_{T}\right)}\right\Vert \le\left\Vert \mathbf{h}_{\left(2:NL_{T}\right)}\right\Vert \le\cdots\le\left\Vert \mathbf{h}_{\left(NL_{T}:NL_{T}\right)}\right\Vert $,
where $\left\Vert \cdot\right\Vert $ denotes the Frobenius norm.
Then the antennas corresponding to $\mathbf{h}_{\left(NL_{T}-s+1:NL_{T}\right)},\cdots,\mathbf{h}_{\left(NL_{T}:NL_{T}\right)}$
are selected to be turned on. 
\end{description}
To feedback the antenna selection information, totally $\log_{2}{NL_{T} \choose s}$
many bits are needed. The corresponding signal model then reduces
to \[
\mathbf{Y}=\sum_{k=1}^{s}\mathbf{h}_{\left(NL_{T}-k+1:NL_{T}\right)}X_{k}+\mathbf{W}.\]

Let $\mathbf{h}_{\left(NL_{T}-k+1:NL_{T}\right)}=n_{k}\bm{\xi}_{k}$
where $n_{k}=\left\Vert \mathbf{h}_{\left(NL_{T}-k+1:NL_{T}\right)}\right\Vert $
and $\bm{\xi}_{k}=\mathbf{h}_{\left(NL_{T}-k+1:NL_{T}\right)}/n_{k}$.
Define $\mathbf{\Xi}:=\left[\bm{\xi}_{1}\cdots\bm{\xi}_{s}\right]$.
Then the sum rate $\mathcal{I}$ is upper bounded by \begin{align}
\mathcal{I} & :=\mathrm{E}_{\mathbf{H}}\left[\log\left|\mathbf{I}_{L_{R}}+\frac{\rho}{s}\mathbf{\Xi}\mathrm{diag}\left[n_{1}^{2},\cdots,n_{s}^{2}\right]\mathbf{\Xi}^{\dagger}\right|\right]\nonumber \\
 & \le\mathrm{E}_{\mathbf{\Xi}}\left[\log\left|\mathbf{I}_{L_{R}}+\frac{\rho}{s}\eta L_{R}\mathbf{\Xi}\mathbf{\Xi}^{\dagger}\right|\right],\label{eq:sum-rate-ub}\end{align}
 where \begin{equation}
\eta:=\frac{1}{sL_{R}}\mathrm{E}_{\mathbf{n}^{2}}\left[\sum_{k=1}^{s}n_{k}^{2}\right],\label{eq:power-efficiency-factor}\end{equation}
 and the inequality comes from the concavity of $\log\left|\cdot\right|$
function \cite{Cover_book91_information_theory} and the fact that
$\mathbf{\Xi}$ and $\mathbf{n}^{2}:=\left[n_{1}^{2}\cdots n_{s}^{2}\right]^{\dagger}$
are independent \cite[Eq. (3.9)]{Edelman_Rao_2005_Random_Matrix_Theory}.
We refer to $\eta$ as \emph{the power efficiency factor} as it describes
the power gain of choosing the strongest antennas against random antenna
selection: if antennas are selected randomly with the total power
constraint increased to $\rho\eta$, the average received signal power
is the same as that of our antenna selection strategy.

Based on the upper bound (\ref{eq:sum-rate-ub}), the sum rate can
be approximately quantified. Note that $\left\Vert \mathbf{h}_{i}\right\Vert $'s
are i.i.d. r.v. with PDF $f\left(x\right)=\frac{1}{\left(L_{R}-1\right)!}x^{L_{R}-1}e^{-x}$.
An application of (\ref{eq:approx-order-statistics}) provides an
accurate approximation of $\eta$. Furthermore, note that $\mathbf{\Xi}\in\mathcal{M}_{L_{R},1}^{\left(s\right)}\left(\mathbb{C}\right)$
is isotropically distributed. Substituting $c=\frac{\rho}{s}\eta L_{R}$
and $\bar{m}=\frac{s}{L_{R}}$ into (\ref{eq:CGM-Shannon-Transform})
estimates the upper bound (\ref{eq:sum-rate-ub}). Simulations in
Section \ref{sec:Simulations-and-Discussion} show that this theoretical
calculation gives a good approximation to the true sum rate.

\subsection{\label{sub:General-Beamforming-Strategy}General Beamforming Strategy}

In this subsection, we propose a power on/off strategy with general
beamforming: the base station selects the strongest users, jointly
quantizes their strongest eigen-channel vectors and broadcasts a common
feedback to all the users; then the on-users transmit along the fedback
beamforming vectors.

\begin{remrk}
We consider this suboptimal strategy for its implementational simplicity
and tractable performance analysis. The user selection is only based
on the Frobenius norm of the channel realization, which does not require
complicated matrix computations. Note that only a few users are chosen
among a large number of total users available and that singular value
decomposition is performed only after user selection in our strategy.
The computation complexity is much lower than a user selection strategy
depending on eigenvalues of the channel matrices. For each selected
user, only the strongest eigen-channel is used. This assumption imposes
a nice symmetric structure and makes analysis tractable. 
\end{remrk}

In particular, for transmission, along with assumptions T1) and T2),
we add one more constraint:

\begin{description}
\item [{T3)}] There is at most one on-beam per user, that is, $s_{i}=0$
or $s_{i}=1$. Note that this also implies that the total number of
on-streams $s$ is the same as the number of on-users. 
\end{description}
For a given channel realization $\mathbf{H}$, we select the on-users
according to 

\begin{description}
\item [{F2)}] User selection criterion. Sort the channel state matrices
$\mathbf{H}_{i}$'s such that $\left\Vert \mathbf{H}_{\left(1:N\right)}\right\Vert \le\left\Vert \mathbf{H}_{\left(2:N\right)}\right\Vert \le\cdots\le\left\Vert \mathbf{H}_{\left(N:N\right)}\right\Vert $,
where $\left\Vert \cdot\right\Vert $ is the Frobenius norm. Then
the users corresponding to $\mathbf{H}_{\left(N-k+1:N\right)},$ $\cdots,\mathbf{H}_{\left(N:N\right)}$
are selected to be turned on.
\end{description}
After selecting the on-users, the base stations also quantizes their
strongest eigen-channel vectors. Consider the singular value decomposition
$\mathbf{H}_{\left(N-k+1:N\right)}=\mathbf{U}_{k}\mathbf{\Lambda}_{k}\mathbf{V}_{k}^{\dagger}$
where the diagonal elements of $\mathbf{\Lambda}_{k}$ are decreasingly
ordered. Let $\mathbf{v}_{k}$ be the column of $\mathbf{V}_{k}$
corresponding to the largest singular value of $\mathbf{\Lambda}_{k}$.
Then the matrix \[
\mathbf{V}:=\left[\mathbf{v}_{1}\cdots\mathbf{v}_{s}\right]\in\mathcal{M}_{L_{T},1}^{\left(s\right)}\left(\mathbb{C}\right),\]
where $\mathcal{M}_{L_{T},1}^{\left(s\right)}\left(\mathbb{C}\right)$
is the set of composite Grassmann matrices (defined in Section \ref{sub:CGMatrix}).
In order to quantize $\mathbf{V}$, the base station constructs a
codebook $\mathcal{B}\subset\mathcal{M}_{L_{T},1}^{\left(s\right)}\left(\mathbb{C}\right)$
with $\left|\mathcal{B}\right|=2^{R_{\mathrm{q}}}$ where $R_{\mathrm{q}}$
is the feedback bits available for eigen-channel vector quantization.
\emph{Note that random codebooks are asymptotically optimal in probability
(Theorem \ref{thm:DRF-CGM-Asymptotics}), we assume that $\mathcal{B}$
is randomly generated from the isotropic distribution.} For a given
eigen-channel vector matrix $\mathbf{V}$, the base station quantizes
$\mathbf{V}$ via the

\begin{description}
\item [{F3)}] Eigen-channel vector quantization function \begin{equation}
\varphi\left(\mathbf{V}\right)=\underset{\mathbf{B}\in\mathcal{B}}{\arg\;\max}\;\sum_{k=1}^{s}\left|\mathbf{v}_{k}^{\dagger}\mathbf{b}_{k}\right|^{2},\label{eq:quantization-fn-eigenchannels}\end{equation}
where $\mathbf{b}_{k}$ is the $k^{\mathrm{th}}$ column of $\mathbf{B}\in\mathcal{B}$.
Indeed, let $P^{\left(m\right)},Q^{\left(m\right)}\in\mathcal{G}_{L_{T},1}^{\left(s\right)}\left(\mathbb{C}\right)$
be the composite planes generated by $\mathbf{V}$ and $\mathbf{B}$
respectively. Then (\ref{eq:quantization-fn-eigenchannels}) is equivalent
to the quantization function on the composite Grassmann manifold defined
in (\ref{eq:quantization-fn-CGM}). 
\end{description}
After quantization, the base station broadcasts the user selection
information (requiring $\log_{2}{N \choose s}$ many feedback bits)
and the index of eigen-channel vector quantization to the users. The
corresponding signal model is now reduced to \begin{align*}
\mathbf{Y} & =\sum_{k=1}^{s}\mathbf{H}_{\left(N-k+1:N\right)}\mathbf{b}_{k}X_{k}+\mathbf{W}\\
 & =\sum_{k=1}^{s}\tilde{\mathbf{h}}_{k}X_{k}+\mathbf{W},\end{align*}
where $\tilde{\mathbf{h}}_{k}:=\mathbf{H}_{\left(N-k+1:N\right)}\mathbf{b}_{k}$
is the equivalent channel for the on-user $k$. 

The point is that the joint quantization (\ref{eq:quantization-fn-eigenchannels})
efficiently employs the feedback resource. It differs from an individual
quantization where each $\mathbf{v}_{k}$ is quantized independently:
separate codebooks $\mathcal{B}_{1},\cdots,\mathcal{B}_{s}$ are constructed
for quantization of $\mathbf{v}_{1},\cdots,\mathbf{v}_{s}$ respectively,
and the quantization function is \[
\varphi^{\prime}\left(\mathbf{V}\right)=\prod_{k=1}^{s}\underset{\mathbf{b}\in\mathcal{B}_{k}}{\arg\;\max}\;\left|\mathbf{v}_{k}^{\dagger}\mathbf{b}\right|\]
where $\prod$ is the Cartesian product. Indeed, individual quantization
is a special case of joint quantization obtained by restricting the
codebook to be a Cartesian product of several individual codebooks.
It is thus obvious that joint quantization achieves a gain tied to
that of vector over scalar quantization.

Certainly the sum rate depends on the codebook. Still, when random
codebooks are considered, it is reasonable to focus upon the ensemble
average sum rate. Let $\tilde{\mathbf{h}}_{k}=n_{k}\bm{\xi}_{k}$
and $\mathbf{\Xi}=\left[\bm{\xi}_{1}\cdots\bm{\xi}_{s}\right]$, where
$n_{k}=\left\Vert \tilde{\mathbf{h}}_{k}\right\Vert $ and $\bm{\xi}_{k}=\tilde{\mathbf{h}}_{k}/n_{k}$.
Then the average sum rate satisfies \begin{align}
\bar{\mathcal{I}}_{\mathrm{rand}} & =\mathrm{E}_{\mathcal{B}}\left[\log\left|\mathbf{I}_{L_{R}}+\frac{\rho}{s}\mathbf{\Xi}\mathrm{diag}\left[n_{1}^{2},\cdots,n_{s}^{2}\right]\mathbf{\Xi}^{\dagger}\right|\right]\nonumber \\
 & \le\mathrm{E}_{\mathbf{\Xi}}\left[\log\left|\mathbf{I}_{L_{R}}+\frac{\rho}{s}\mathrm{E}_{\mathcal{B}}\left[\eta\right]L_{R}\mathbf{\Xi\Xi}^{\dagger}\right|\right],\label{eq:average-sum-rate-ub}\end{align}
 where $\eta$ is defined in (\ref{eq:power-efficiency-factor}).
The inequality in the second line follows from Jensen's inequality
and the next fact.

\begin{thm}
\label{thm:prob-properties}$\bm{\xi}_{k}$'s $1\le k\le s$ are independent
and isotropically distributed. Furthermore, $\bm{\xi}_{k}$'s are
independent of $n_{k}$'s. 
\end{thm}
\begin{proof}
Consider the singular value decomposition of a standard Gaussian matrix
$\mathbf{H}=\mathbf{U}\mathbf{\Lambda}\mathbf{V}^{\dagger}$. It is
well known that $\mathbf{U}$ and $\mathbf{V}$ are independent and
isotropically distributed, and both of them are independent of $\mathbf{\Lambda}$
\cite[Eq. (3.9)]{Edelman_Rao_2005_Random_Matrix_Theory}. Now let
$\mathbf{U}_{k}\mathbf{\Lambda}_{k}\mathbf{V}_{k}^{\dagger}$ be the
singular value decomposition of $\mathbf{H}_{\left(N-k+1:N\right)}$
$1\le k\le s$. Since we choose users only according to their Frobenius
norms, the choice of $\mathbf{H}_{\left(N-k+1:N\right)}$ only depends
on $\mathbf{\Lambda}$ but is independent of $\mathbf{U}_{k}$ and
$\mathbf{V}_{k}$. The independence among $\mathbf{U}_{k}$, $\mathbf{V}_{k}$
and $\mathbf{\Lambda}_{k}$ still holds. Note that the equivalent
channel vector $\tilde{\mathbf{h}}_{k}=\mathbf{U}_{k}\mathbf{\Lambda}_{k}\mathbf{V}_{k}^{\dagger}\mathbf{b}_{k}=\mathbf{U}_{k}\hat{\bm{\xi}}_{k}n_{k}$
where $\mathbf{\Lambda}_{k}\mathbf{V}_{k}^{\dagger}\mathbf{b}_{k}=\hat{\bm{\xi}}_{k}n_{k}$.
Since $\mathbf{b}_{k}$ depends only on $\mathbf{V}_{k}$, $\mathbf{U}_{k}$
is independent of $\hat{\bm{\xi}}_{k}$. Thus $\bm{\xi}_{k}=\mathbf{U}_{k}\hat{\bm{\xi}}_{k}$
is isotropically distributed \cite{James_54_Normal_Multivariate_Analysis_Orthogonal_Group}.
Now the fact that $\mathbf{U}_{k}$'s are independent across $k$
implies that $\bm{\xi}_{k}$'s are independent across $k$ \cite{James_54_Normal_Multivariate_Analysis_Orthogonal_Group}. 

Next realize that $n_{k}$ is only a function of $\mathbf{\Lambda}_{k}$
and $\mathbf{V}_{k}$, both of which are independent of $\mathbf{U}_{k}$.
$\mathbf{U}_{k}$'s are independent of $n_{k}$'s (and isotropically
distributed). It follows that $\bm{\xi}_{k}$'s and $n_{k}$'s are
independent \cite{James_54_Normal_Multivariate_Analysis_Orthogonal_Group}. 
\end{proof}

The calculation of $\mathrm{E}_{\mathcal{B}}\left[\eta\right]$ proceeds
as follows. To simplify notation, let $\mathbf{H}_{\left(k\right)}=\mathbf{H}_{\left(N-k+1:N\right)}$
and $n_{\left(k\right)}^{2}=\left\Vert \mathbf{H}_{\left(k\right)}\right\Vert ^{2}$.
Let $\bar{n}_{\left(\cdot\right)}^{2}=\frac{1}{s}\sum_{k=1}^{s}\mathrm{E}\left[n_{\left(k\right)}^{2}\right]$.
Let $\lambda_{k,j}$ ($1\le k\le s$ and $1\le j\le L_{R}$) be the
decreasingly ordered eigenvalues of $\mathbf{H}_{\left(k\right)}\mathbf{H}_{\left(k\right)}^{\dagger}$,
and $\zeta_{j}=\mathrm{E}\left[\left.\lambda_{k,j}\right|n_{\left(k\right)}^{2}=1\right]$,
defined in Lemma \ref{lem:conditional-expectation-Wishart}. For a
quantization codebook $\mathcal{B}$, let $\mathbf{v}_{k}$ be the
$k^{\mathrm{th}}$ column of $\mathbf{V}$, and $\mathbf{b}_{k}$
be the $k^{\mathrm{th}}$ column of $\mathbf{B}=\varphi\left(\mathbf{V}\right)\in\mathcal{B}$.
Define \begin{align}
\gamma & :=\mathrm{E}_{\mathcal{B},\mathbf{V}}\left[\sum_{k=1}^{s}\left|\mathbf{v}_{k}^{\dagger}\mathbf{b}_{k}\right|^{2}\right].\label{eq:def-gamma}\end{align}
$\mathrm{E}_{\mathcal{B}}\left[\eta\right]$ is a function of $\gamma$.

\begin{thm}
\label{thm:n_bar}Let the random codebook $\mathcal{B}$ follows the
isotropic distribution. Then\begin{align}
\mathrm{E}_{\mathcal{B}}\left[\eta\right] & =\mathrm{E}_{\mathcal{B},\mathbf{V}}\left[\sum_{k=1}^{s}\left|\mathbf{v}_{k}^{\dagger}\mathbf{b}_{k}\right|^{2}\right]\nonumber \\
 & =L_{R}\left(\frac{\gamma}{s}\zeta_{1}+\frac{s-\gamma}{s}\frac{1-\zeta_{1}}{L_{T}-1}\right)\bar{n}_{\left(\cdot\right)}^{2}.\label{eq:power-efficiency-calculation}\end{align}

\end{thm}

The proof is contained in Appendix \ref{sub:Proof-of-Theorem-n_bar}.

To make use of this formula, the constant $\zeta_{1}$ can be well
approximated by $\bar{\zeta}_{1/L_{R}}$ using our results in Section
\ref{sub:Conditioned-Eigen}, and $\bar{n}_{\left(\cdot\right)}^{2}$
can be estimated by (\ref{eq:approx-order-statistics}). Let $R_{q}$
be the quantization rate on eigen-channel vector quantization. As
a function of $R_{q}$, an approximation of $\gamma$ is provided
at the end of Section \ref{sub:CGManifold}. Put together we have
our estimate of $\mathrm{E}_{\mathcal{B}}\left[\eta\right]$. And
to estimate the average sum rate, we only need to substitute the value
of $\mathrm{E}_{\mathcal{B}}\left[\eta\right]$ into the bound (\ref{eq:average-sum-rate-ub})
and then evaluate it via (\ref{eq:CGM-Shannon-Transform}).

\subsection{\label{sub:Comments}Comments}

\subsubsection{Choice of $s$}

The number of on-beams $s$ should be chosen to maximize the sum rate
keeping in mind that it is a function of SNR $\rho$. Given that our
proved bound accurately approximates the sum rate (when $s\ll N$
and $R_{\mathrm{q}}$ are large enough), the optimal number of on-beams
$s^{*}$ can be found by a simple search.

\subsubsection{Antenna Selection and General Beamforming}

The antenna selection can be viewed as a special case of general beamforming
where a beamforming vector has a particular structure - it must be
a column of the identity matrix. Note that general beamforming requires
total feedback rate $\log_{2}{N \choose s}+R_{\mathrm{q}}$ bits while
antenna selection needs $\log_{2}{NL_{T} \choose s}=\log_{2}{N \choose s}+s\log_{2}L_{T}+O\left(\frac{1}{N}\right)$
bits for feedback. Antenna selection can be viewed as general beamforming
with $R_{\mathrm{q}}=s\log_{2}L_{T}$. One difference between antenna
selection and general beamforming is that antenna selection does not
assume one on-beam per on-user (Assumption T3)). In antenna selection,
multiple antennas corresponding to the same user can be turned on
simultaneously. As a result, the sum rate achieved by antenna selection
is expected to be better than that of general beamforming with $R_{\mathrm{q}}=s\log_{2}L_{T}$.
This is supported in our simulations.

\section{\label{sec:Simulations-and-Discussion}Simulations and Discussion}

Simulations for antenna selection and general beamforming strategies
are presented in Fig. \ref{cap:Antenna-Selection} and \ref{cap:Beamforming}
respectively. Fig. \ref{cap:Antenna-Selection} shows the sum rate
of antenna selection versus SNR. The circles are simulated sum rates,
the solid lines are simulated upper bounds (\ref{eq:sum-rate-ub}),
the plus markers are the sum rates calculated by theoretical approximation,
and the dotted lines are the sum rates corresponding to the case where
there is no CSIT at all. In the simulations, the value of $s$ is
chosen to maximize the sum rate according to our theoretical analysis.
Fig. \ref{cap:Beamforming} illustrates how the sum rate increases
as the eigen-channel vectors quantization rate $R_{\mathrm{q}}$ increases.
Here, the $s$ is fixed to be $4$. The dash-dot lines denote perfect
beamforming, which corresponds to $R_{\mathrm{q}}=+\infty$. The circles
are for our proposed joint strategy, the solid lines are simulated
upper bounds (\ref{eq:sum-rate-ub}), the up-triangles are for antenna
selection and the down-triangles are for individual eigen-channel
vectors quantization (recall the detailed discussion in Section \ref{sub:General-Beamforming-Strategy}).
We observe the following.

\begin{itemize}
\item The upper bounds (\ref{eq:sum-rate-ub}) and (\ref{eq:average-sum-rate-ub})
appear to be good approximations to the sum rate. 
\item The sum rate increases as the number of users $N$ increases. Fig.
\ref{cap:Antenna-Selection} compares the $N=32$ and $N=256$ cases.
Our analysis bears out that increasing $N$ results in an increase
in the equivalent channel norms according to extreme order statistics.
The power efficiency factor increases and therefore the sum rate performance
improves. 
\item The loss due to eigen-channel vector quantization decreases exponentially
as $R_{\mathrm{q}}$ increases. According to Theorem \ref{thm:DRF-CGM},
the decay rate is $\frac{1}{s\left(L_{T}-1\right)}R_{\mathrm{q}}$.
When $L_{T}$ is not large (which is often true in practice), a relatively
small $R_{\mathrm{q}}$ may be good enough. In Fig. \ref{cap:Beamforming},
as $L_{T}=2$ and $s=4$, $R_{\mathrm{q}}=12$ bits is almost as good
as perfect beamforming. 
\item Our proposed joint strategy achieves better performance than individual
quantization. Note that the effect of eigen-channel vectors quantization
is characterized by a single parameter $\gamma$. Joint quantization
yields larger $\gamma$, larger power efficiency factor, and therefore
better performance. 
\item Antenna selection is only slightly better than general beamforming
with $R_{\mathrm{q}}=s\log_{2}L_{T}$. As has been discussed in Section
\ref{sub:Comments}, the performance improvement is due to excluding
the assumption T3). 
\end{itemize}
\begin{figure}[tbh]
\begin{minipage}[c][1\totalheight][t]{0.5\columnwidth}%
\includegraphics[clip,scale=0.6]{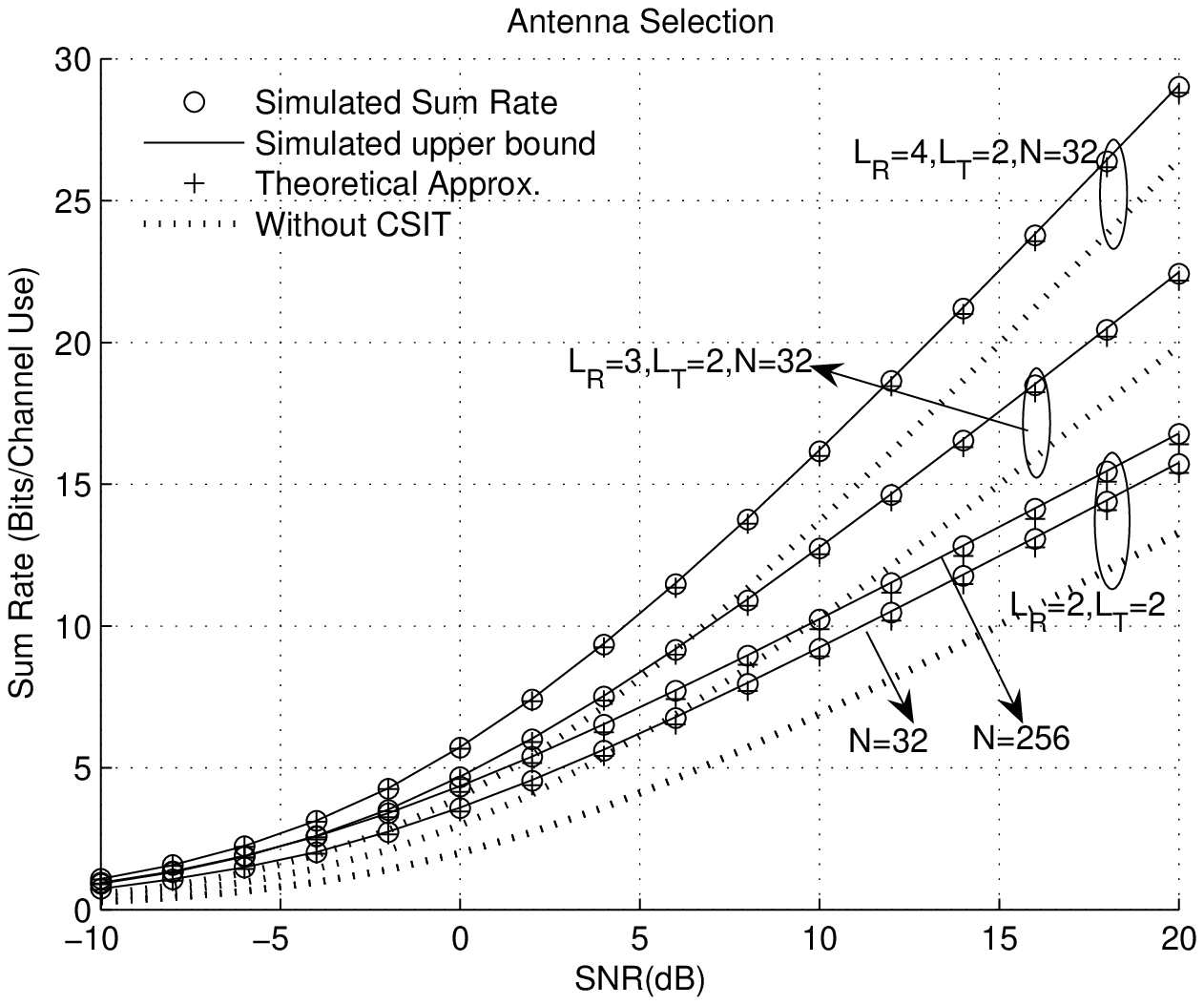}

\caption{\label{cap:Antenna-Selection}Antenna Selection: Sum Rate versus SNR.}
\end{minipage}%
\begin{minipage}[c][1\totalheight][t]{0.5\columnwidth}%
\includegraphics[clip,scale=0.6]{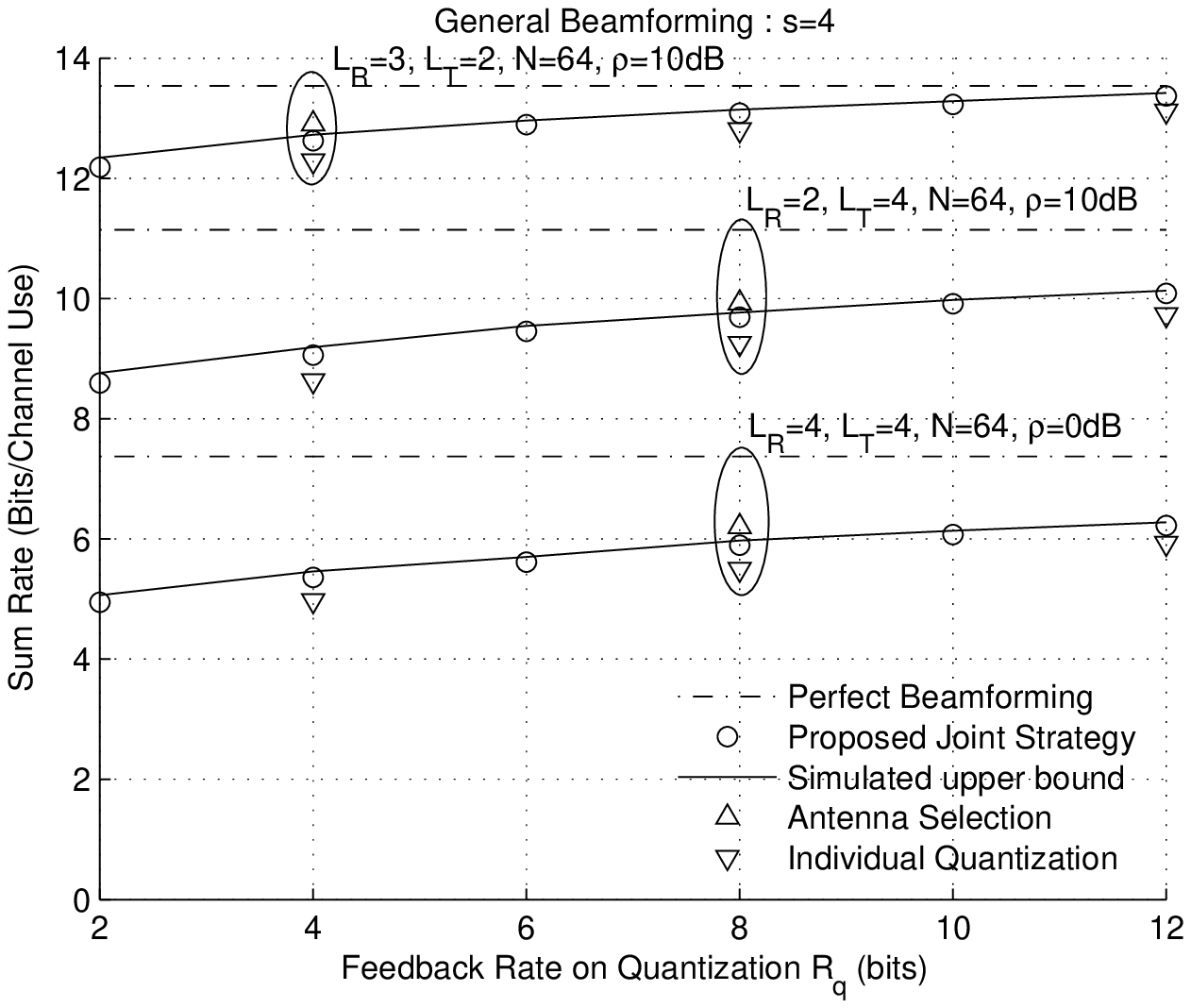}

\caption{\label{cap:Beamforming}General Beamforming: Sum Rate versus $R_{\mathrm{q}}$.}
\end{minipage}
\end{figure}

\section{\label{sec:Conclusion}Conclusion}

This paper proposes a joint quantization and feedback strategy for
multiaccess MIMO systems with finite rate feedback. The effect of
user choice is analyzed by extreme order statistics and the effect
of eigen-channel vector quantization is quantified by analysis on
the composite Grassmann manifold. By asymptotic random matrix theory,
the sum rate is well approximated. Due to its simple implementation
and solid performance analysis, the proposed scheme provides a benchmark
for multiaccess MIMO systems with finite rate feedback. 

\appendix

\subsection{\label{sub:Random-Matrix-Theory}Random Matrix Theory}

Let $\mathbf{H}\in\mathbb{L}^{n\times m}$ be a standard Gaussian
random matrix, where $\mathbb{L}$ is either $\mathbb{R}$ or $\mathbb{C}$.
Let $\lambda_{1},\cdots,\lambda_{n}$ be the $n$ singular values
of $\frac{1}{m}\mathbf{H}\mathbf{H}^{\dagger}$. Define the empirical
distribution of the singular values \[
\mu_{n,\bm{\lambda}}\left(\lambda\right)\triangleq\frac{1}{n}\left|\left\{ j:\;\lambda_{j}\le\lambda\right\} \right|.\]
 As $n,m\rightarrow\infty$ with $\frac{m}{n}\rightarrow\bar{m}\in\mathbb{R}^{+}$,
the empirical measure converges to the Mar\u{c}enko-Pastur law \begin{equation}
d\mu_{\lambda}=\left(\left(1-\bar{m}\right)^{+}\delta\left(\lambda\right)+\frac{\bar{m}\sqrt{\left(\lambda-\lambda^{-}\right)^{+}\left(\lambda^{+}-\lambda\right)^{+}}}{2\pi\lambda}\right)\, d\lambda\label{eq:spectrum-pdf}\end{equation}
 almost surely, where $\lambda^{\pm}=\left(1\pm\sqrt{\frac{1}{\bar{m}}}\right)^{2}$
and $\left(x\right)^{+}=\max\left(x,0\right)$ (A good reference for
this type of result is \cite[Eq. (1.10)]{Verdu_random_matrix_theory_wireless_communications}).
Define\[
\lambda_{t}^{-}\triangleq\begin{cases}
0 & \mathrm{if}\;\beta\ge1\\
\lambda^{-} & \mathrm{if}\;\beta<1\end{cases}.\]
 Consider as well a linear spectral statistic \[
g\left(\frac{1}{m}\mathbf{H}\mathbf{H}^{\dagger}\right)=\frac{1}{n}\sum_{i=1}^{n}g\left(\lambda_{i}\right).\]
 If $g$ is Lipschitz on $\left[\lambda_{t}^{-},\lambda^{+}\right]$,
then we also have that \[
\underset{\left(n,m\right)\rightarrow\infty}{\lim}g\left(\frac{1}{m}\mathbf{H}\mathbf{H}^{\dagger}\right)=\int g\left(\lambda\right)d\mu_{\lambda}\]
 almost surely, see for example \cite{Guionnet_Zeitouni_2000} for
a modern approach.

The asymptotic properties of the maximum eigenvalue will figure into
our analysis. Denote the largest eigenvalue by $\lambda_{1}$. 

\begin{prop}
\label{prop:Asymptotic-Lambda-Max}Let $n,m\rightarrow\infty$ linearly
with $\frac{m}{n}\rightarrow\bar{m}\in\mathbb{R}^{+}$. 
\begin{enumerate}
\item $\lambda_{1}\rightarrow\lambda^{+}$ almost surely. 
\item All moments of $\lambda_{1}$ also converge. 
\end{enumerate}
\end{prop}
The almost sure convergence goes back to \cite{Bai_Yin_Krishnaih_1988,Bai_Yin_1993}.
The convergence of moments is implied by the tail estimates in \cite{Ledoux_2005_Orthogonal_Polynomials}.
A direct application of this proposition is that for $\forall A_{n}\subset\mathbb{R}^{n}$
such that $\mu_{n,\bm{\lambda}}\left(A_{n}\right)\rightarrow0$, $\mathrm{E}_{\bm{\lambda}}\left[\lambda_{1},\; A_{n}\right]\rightarrow0$.

\begin{thm}
\label{thm:truncate-function-RMT}Let $\mathbf{H}\in\mathbb{L}^{n\times m}$
($\mathbb{L}=\mathbb{R}/\mathbb{C}$) be standard Gaussian matrix
and $\lambda_{i}$ be the $i^{\mathrm{th}}$ largest eigenvalue of
$\frac{1}{m}\mathbf{H}\mathbf{H}^{\dagger}$. 
\begin{enumerate}
\item Let $g\left(\lambda\right)=f\left(\lambda\right)\cdot\chi_{\left[a,\lambda^{+}\right]}\left(\lambda\right)$
for some $a<\lambda^{+}$ where $f\left(\lambda\right)$ is Lipschitz
continuous on $\left[\lambda^{-},\lambda^{+}\right]$ and $\chi_{\left[a,\lambda^{+}\right]}\left(\lambda\right)$
is the indicator function on the set $\left[a,\lambda^{+}\right]$,
then as $n,m\rightarrow\infty$ with $\frac{m}{n}\rightarrow\bar{m}\in\mathbb{R}^{+}$,
\[
\underset{\left(n,m\right)\rightarrow\infty}{\lim}\int g\left(\lambda\right)\cdot d\mu_{n,\bm{\lambda}}\left(\lambda\right)=\int g\left(\lambda\right)\cdot d\mu_{\lambda}\]
almost surely and \[
\underset{\left(n,m\right)\rightarrow\infty}{\lim}\mathrm{E}_{\bm{\lambda}}\left[\frac{1}{n}\sum_{i=1}^{n}g\left(\lambda_{i}\right)\right]=\int g\left(\lambda\right)\cdot d\mu_{\lambda}.\]

\item For $\forall a\in\left(\lambda_{t}^{-},\lambda^{+}\right)$, \[
\mathrm{E}_{\bm{\lambda}}\left[\frac{1}{n}\left|\left\{ \lambda_{i}:\;\lambda_{i}\ge a\right\} \right|\right]=\int_{a}^{\lambda^{+}}d\mu_{\lambda}.\]

\item For $\forall\tau\in\left(0,\min\left(1,\bar{m}\right)\right)$, \[
\underset{\left(n,m\right)\rightarrow\infty}{\lim}\mathrm{E}\left[\frac{1}{n}\left(\sum_{1\le i\le n\tau}\lambda_{i}\right)\right]=\int_{a}^{\lambda^{+}}\lambda\cdot d\mu_{\lambda},\]
where $a\in\left(\lambda^{-},\lambda^{+}\right)$ satisfies \[
\tau=\int_{a}^{\lambda^{+}}d\mu_{\lambda}.\]

\end{enumerate}
\end{thm}
\begin{proof}
\begin{enumerate}
\item Though $g\left(\lambda\right)$ is not Lipschitz continuous on $\left[\lambda_{t}^{-},\lambda^{+}\right]$,
we are able to construct sequences of Lipschitz functions $g_{k}^{+}\left(\lambda\right)$
and $g_{k}^{-}\left(\lambda\right)$ such that $g_{k}^{\pm}\left(\lambda\right)$'s
are Lipschitz continuous on $\left[\lambda_{t}^{-},\lambda^{+}\right]$
for all $k$, $g_{k}^{+}\left(\lambda\right)\ge g\left(\lambda\right)$
and $g_{k}^{-}\left(\lambda\right)\le g\left(\lambda\right)$ for
$\lambda\in\left[\lambda_{t}^{-},\lambda^{+}\right]$, and $g_{k}^{\pm}\left(\lambda\right)\rightarrow g\left(\lambda\right)$
pointwisely as $k\rightarrow\infty$. Due to their Lipschitz continuity,
$g_{k}^{\pm}\left(\lambda\right)$'s are integrable with respect to
$\mu_{\lambda}$. Then we have \begin{align*}
\underset{k\rightarrow\infty}{\lim}\underset{\left(n,m\right)\rightarrow\infty}{\lim}\int g_{k}^{-}\left(\lambda\right)\cdot d\mu_{n,\bm{\lambda}}\left(\lambda\right) & \le\underset{\left(n,m\right)\rightarrow\infty}{\lim}\int g\left(\lambda\right)\cdot d\mu_{n,\bm{\lambda}}\left(\lambda\right)\\
 & \le\underset{k\rightarrow\infty}{\lim}\underset{\left(n,m\right)\rightarrow\infty}{\lim}\int g_{k}^{+}\left(\lambda\right)\cdot d\mu_{n,\bm{\lambda}}\left(\lambda\right),\end{align*}
 while \[
\underset{k\rightarrow\infty}{\lim}\underset{\left(n,m\right)\rightarrow\infty}{\lim}\int g_{k}^{-}\left(\lambda\right)\cdot d\mu_{n,\bm{\lambda}}\left(\lambda\right)=\underset{k\rightarrow\infty}{\lim}\int g_{k}^{-}\left(\lambda\right)\cdot d\mu\left(\lambda\right)=\int g\left(\lambda\right)\cdot d\mu\left(\lambda\right)\]
 and \[
\underset{k\rightarrow\infty}{\lim}\underset{\left(n,m\right)\rightarrow\infty}{\lim}\int g_{k}^{+}\left(\lambda\right)\cdot d\mu_{n,\bm{\lambda}}\left(\lambda\right)=\int g\left(\lambda\right)\cdot d\mu\left(\lambda\right)\]
 almost surely. This proves the almost sure statement, and the convergence
of the expectation follows from dominated convergence.
\item follows from the first part upon setting $g\left(\lambda\right)=\chi_{\left[a,\lambda^{+}\right]}\left(\lambda\right)$. 
\item Since $a\in\left(\lambda^{-},\lambda^{+}\right)$, there exists an
$\epsilon>0$ such that $\left(a-\epsilon,a+\epsilon\right)\subset\left(\lambda^{-},\lambda^{+}\right)$.
For any $\delta>0$, define the events\[
A_{n,a+\epsilon}=\left\{ \bm{\lambda}:\;\frac{\left|\left\{ \lambda_{i}:\;\lambda_{i}\ge a+\epsilon\right\} \right|}{n}<\tau\right\} ,\]
\[
A_{n,a-\epsilon}=\left\{ \bm{\lambda}:\;\frac{\left|\left\{ \lambda_{i}:\;\lambda_{i}\ge a-\epsilon\right\} \right|}{n}>\tau\right\} ,\]
\[
B_{n,a+\epsilon,\delta}=\left\{ \bm{\lambda}:\;\left|\frac{1}{n}\sum_{\lambda_{i}\ge a+\epsilon}\lambda_{i}-\int_{a+\epsilon}^{\lambda^{+}}\lambda\cdot d\mu_{\lambda}\right|<\delta\right\} ,\]
and \[
B_{n,a-\epsilon,\delta}=\left\{ \bm{\lambda}:\;\left|\frac{1}{n}\sum_{\lambda_{i}\ge a-\epsilon}\lambda_{i}-\int_{a-\epsilon}^{\lambda^{+}}\lambda\cdot d\mu_{\lambda}\right|<\delta\right\} .\]
According to the first part of this theorem, it can be verified that
$\forall\epsilon>0$, as $\left(n,m\right)\rightarrow\infty$, $\mu_{n,\bm{\lambda}}\left(A_{n,a+\epsilon}\right)\rightarrow1$,
$\mu_{n,\bm{\lambda}}\left(A_{n,a-\epsilon}\right)\rightarrow1$,
$\mu_{n,\bm{\lambda}}\left(B_{n,a+\epsilon,\delta}\right)\rightarrow1$,
and $\mu_{n,\bm{\lambda}}\left(B_{n,a-\epsilon,\delta}\right)\rightarrow1$.
Then for sufficiently large $n$, \begin{align}
 & \mathrm{E}_{\bm{\lambda}}\left[\frac{1}{n}\sum_{i\le n\tau}\lambda_{i}\right]\nonumber \\
 & \ge\mathrm{E}_{\bm{\lambda}}\left[\frac{1}{n}\sum_{i\le n\tau}\lambda_{i},\; A_{n,a+\epsilon}\cap B_{n,a+\epsilon,\delta}\right]\nonumber \\
 & \overset{\left(a\right)}{\ge}\mathrm{E}_{\bm{\lambda}}\left[\frac{1}{n}\sum_{\lambda_{i}\ge a+\epsilon}\lambda_{i},\; A_{n,a+\epsilon}\cap B_{n,a+\epsilon,\delta}\right]\nonumber \\
 & \overset{\left(b\right)}{\ge}\mathrm{E}_{\bm{\lambda}}\left[\int_{a+\epsilon}^{\lambda^{+}}\lambda\cdot d\mu_{\lambda}-\delta,\; A_{n,a+\epsilon}\cap B_{n,a+\epsilon,\delta}\right]\nonumber \\
 & =\left(\int_{a+\epsilon}^{\lambda^{+}}\lambda\cdot d\mu_{\lambda}-\delta\right)\mu_{n,\bm{\lambda}}\left(A_{n,a+\epsilon}\cap B_{n,a+\epsilon,\delta}\right)\nonumber \\
 & \ge\left(\int_{a+\epsilon}^{\lambda^{+}}\lambda\cdot d\mu_{\lambda}-\delta\right)\left(1-\delta\right),\label{eq:E-eigenvalue-lb}\end{align}
where $\mathrm{E}\left[\cdot,A\right]$ denotes the expectation operation
on the measurable set $A$, $\left(a\right)$ and $\left(b\right)$
follow from the definition of $A_{n,a+\epsilon}$ and $B_{n,a+\epsilon,\delta}$
respectively. Similarly, when $n$ is large enough, \begin{align}
 & \mathrm{E}_{\bm{\lambda}}\left[\frac{1}{n}\sum_{i\le n\tau}\lambda_{i}\right]\nonumber \\
 & \le\mathrm{E}_{\bm{\lambda}}\left[\frac{1}{n}\sum_{i\le n\tau}\lambda_{i},\; A_{n,a-\epsilon}\right]+\mathrm{E}_{\bm{\lambda}}\left[\lambda_{1},\; A_{n,a-\epsilon}^{c}\right]\nonumber \\
 & \overset{\left(c\right)}{\le}\mathrm{E}_{\bm{\lambda}}\left[\frac{1}{n}\sum_{\lambda_{i}\ge a-\epsilon}\lambda_{i},\; A_{n,a-\epsilon}\right]+\delta\nonumber \\
 & \le\mathrm{E}_{\bm{\lambda}}\left[\int_{a-\epsilon}^{\lambda^{+}}\lambda\cdot d\mu_{\lambda}+\delta,\; A_{n,a-\epsilon}\cap B_{n,a-\epsilon,\delta}\right]\nonumber \\
 & \quad\quad+\mathrm{E}_{\bm{\lambda}}\left[\lambda_{1},\; A_{n,a-\epsilon}\cap B_{n,a-\epsilon,\delta}^{c}\right]+\delta\nonumber \\
 & \overset{\left(d\right)}{\le}\left(\int_{a-\epsilon}^{\lambda^{+}}\lambda\cdot d\mu_{\lambda}+\delta\right)\mu_{n,\bm{\lambda}}\left(A_{n,a-\epsilon}\cap B_{n,a-\epsilon,\delta}\right)+2\delta\nonumber \\
 & \le\int_{a-\epsilon}^{\lambda^{+}}\lambda\cdot d\mu_{\lambda}+3\delta,\label{eq:E-eigenvalue-ub}\end{align}
where $\left(c\right)$ and $\left(d\right)$ are an application of
Proposition \ref{prop:Asymptotic-Lambda-Max}. Now let $\delta\downarrow0$
and then $\epsilon\downarrow0$. Then we have proved that \[
\underset{\left(n,m\right)\rightarrow\infty}{\lim}\mathrm{E}\left[\frac{1}{n}\left(\sum_{1\le i\le n\tau}\lambda_{i}\right)\right]=\int_{a}^{\lambda^{+}}\lambda\cdot d\mu_{\lambda}.\]

\end{enumerate}
\end{proof}

\subsection{\label{sub:Proof-of-Extreme-Order-Statistics}Proof of Lemma \ref{lem:Expectation-extreme-chi2}}

As the first step, we compute  the asymptotic distribution and expectation
of $X_{\left(n:n\right)}$. It can be verified that \[
1-F_{X}\left(y\right)=\int_{y}^{+\infty}f_{X}\left(x\right)dx=e^{-y}\left(\sum_{i=0}^{L-1}\frac{1}{i!}y^{i}\right),\]
and for $\forall a>0$,\begin{align*}
 & \int_{a}^{+\infty}1-F_{X}\left(y\right)dy\\
 & =e^{-a}\left(\sum_{i=0}^{L-1}\frac{1}{i!}a^{i}+\sum_{i=0}^{L-2}\frac{1}{i!}a^{i}+\cdots+\sum_{i=0}^{0}\frac{1}{i!}a^{i}\right)\\
 & =e^{-a}\left(\sum_{i=0}^{L-1}\frac{L-i}{i!}a^{i}\right).\end{align*}
For $0<t<+\infty$, define \[
R\left(t\right)=\frac{\int_{t}^{+\infty}\left(1-F_{X}\left(y\right)\right)dy}{1-F_{X}\left(t\right)}.\]
Then \begin{equation}
\underset{t\rightarrow+\infty}{\lim}R\left(t\right)=\underset{t\rightarrow+\infty}{\lim}\frac{\sum_{i=0}^{L-1}\frac{L-i}{i!}t^{i}}{\sum_{i=0}^{L-1}\frac{1}{i!}t^{i}}=1.\label{eq:_R(t)_limit}\end{equation}
Now let \[
a_{n}=\inf\left\{ x:1-F_{X}\left(x\right)\leq\frac{1}{n}\right\} ,\]
and \[
b_{n}=R\left(a_{n}\right)=\frac{\sum_{i=0}^{L-1}\frac{L-i}{i!}a_{n}^{i}}{\sum_{i=0}^{L-1}\frac{1}{i!}a_{n}^{i}}.\]
It can be verified that $a_{n}\overset{n\rightarrow\infty}{\longrightarrow}+\infty$,
and that $b_{n}\overset{n\rightarrow\infty}{\longrightarrow}1$ by
(\ref{eq:_R(t)_limit}). Furthermore,\begin{align}
 & \underset{n\rightarrow\infty}{\lim}n\left[1-F_{X}\left(a_{n}+xb_{n}\right)\right]\nonumber \\
 & =\underset{n\rightarrow\infty}{\lim}\frac{1-F_{X}\left(a_{n}+xb_{n}\right)}{1-F_{X}\left(a_{n}\right)}\nonumber \\
 & =\underset{n\rightarrow\infty}{\lim}e^{-xb_{n}}\frac{\sum_{i=0}^{L-1}\frac{1}{i!}\left(a_{n}+xb_{n}\right)^{i}}{\sum_{i=0}^{L-1}\frac{1}{i!}\left(a_{n}\right)^{i}}\nonumber \\
 & =e^{-x}.\label{eq:_exponential_element}\end{align}
Therefore, for all $x\in\mathbb{R}$ and sufficiently large $n$,
\begin{align*}
 & P\left(X_{\left(n:n\right)}<a_{n}+b_{n}x\right)\\
 & =\left[1-\frac{1}{n}n\left(1-F_{X}\left(a_{n}+b_{n}x\right)\right)\right]^{n}\\
 & =\exp\left(n\cdot\log\left(1-\frac{1}{n}e^{-x}\left(1+o\left(1\right)\right)\right)\right)\\
 & =\exp\left(-e^{-x}\left(1+o\left(1\right)\right)\right)\\
 & \overset{n\rightarrow\infty}{\longrightarrow}\exp\left(-e^{-x}\right).\end{align*}
This identifies the limiting distribution, and the tail is of sufficient
decay to conclude that \[
\underset{n\rightarrow+\infty}{\lim}\mathrm{E}\left[\frac{X_{\left(n:n\right)}-a_{n}}{b_{n}}\right]=\int_{-\infty}^{+\infty}xde^{-e^{-x}}:=\mu_{1}.\]

Given the law of the first maxima $X_{\left(n:n\right)}$, the distribution
and the expectation of the $k^{\mathrm{th}}$ maxima follow easily.
With $z_{n}=a_{n}+b_{n}x$, \[
P\left(X_{\left(n-k+1:n\right)}\le z_{n}\right)=\sum_{t=0}^{k-1}{n \choose t}\left(1-F_{X}\left(z_{n}\right)\right)^{t}F_{X}^{n-t}\left(z_{n}\right).\]
According to (\ref{eq:_exponential_element}), ${n \choose t}\left(1-F_{X}\left(z_{n}\right)\right)^{t}\overset{n\rightarrow\infty}{\longrightarrow}\frac{1}{t!}e^{-tx}$
and $F_{X}^{n-t}\left(z_{n}\right)\overset{n\rightarrow\infty}{\longrightarrow}e^{-e^{-x}}$.
Thus\[
P\left(\frac{X_{\left(n-k+1:n\right)}-a_{n}}{b_{n}}\le x\right)\overset{n\rightarrow\infty}{\longrightarrow}\exp\left(-e^{-x}\right)\sum_{t=0}^{k-1}\frac{1}{t!}e^{-tx}.\]
Denote it by $H_{k}\left(x\right)$. The corresponding PDF is given
by \begin{equation}
h_{k}\left(x\right)=H_{k}^{'}\left(x\right)=e^{-e^{-x}}\frac{1}{\left(k-1\right)!}e^{-kx}.\label{eq:h_k_(x)}\end{equation}
Define $\mu_{k}=\int_{-\infty}^{+\infty}xh_{k}\left(x\right)dx$.
It can be verified that Evaluating $\mu_{k}^{x}$ gives an iterative
formula\begin{align*}
\mu_{k} & =\frac{1}{\left(k-1\right)!}\int_{-\infty}^{+\infty}xe^{-kx}e^{-e^{-x}}dx\\
 & =0-\frac{1}{\left(k-1\right)!}\int_{-\infty}^{+\infty}e^{-e^{-x}}d\left(xe^{-(k-1)x}\right)\\
 & =\frac{1}{\left(k-2\right)!}\int_{-\infty}^{+\infty}xe^{-\left(k-1\right)x}e^{-e^{-x}}dx\\
 & \quad\quad-\frac{1}{\left(k-1\right)!}\int_{-\infty}^{+\infty}e^{-\left(k-1\right)x}e^{-e^{-x}}dx\\
 & =\mu_{k-1}^{x}-\frac{1}{k-1},\end{align*}
where the last step follows the fact that $\frac{1}{\left(k-2\right)!}e^{-\left(k-1\right)}\exp\left(-e^{-x}\right)$
is the asymptotic pdf of $\left(k-1\right)^{\mathrm{th}}$ maxima.
Therefore,\[
\underset{n\rightarrow+\infty}{\lim}\mathrm{E}\left[\frac{X_{\left(n-k+1:n\right)}-a_{n}}{b_{n}}\right]=\mu_{k}=\mu_{1}-\sum_{i=1}^{k-1}\frac{1}{i}.\]
and so also,\[
\underset{n\rightarrow+\infty}{\lim}\mathrm{E}\left[\frac{\sum_{k=1}^{s}X_{\left(n-k+1:n\right)}-sa_{n}}{b_{n}}\right]=\sum_{k=1}^{s}\mu_{k}=s\mu_{1}-\sum_{i=1}^{s}\frac{s-i}{i}.\]

\subsection{\label{sub:proof_DRF_composite_GM}Proof of Theorem \ref{thm:DRF-CGM}}

The proof of Theorem \ref{thm:DRF-CGM} is similar to that of Theorem
2 in our earlier paper \cite{Dai_IT2008_Quantization_Grassmannian_manifold};
the difference being that the composite Grassmann manifold is of interest
here while the {}``single'' Grassmann manifold is the focus in that
work. The key step of this proof is the volume calculation of a small
ball in the composite Grassmann manifold. Given the volume formula,
the upper and lower bounds follow from the exact arguments in \cite{Dai_IT2008_Quantization_Grassmannian_manifold}.

A metric ball in $\mathcal{G}_{n,p}^{\left(m\right)}\left(\mathbb{L}\right)$
centered at $P^{\left(m\right)}\in\mathcal{G}_{n,p}^{\left(m\right)}\left(\mathbb{L}\right)$
with radius $\delta\ge0$ is defined as \[
B_{P^{\left(m\right)}}\left(\delta\right):=\left\{ Q^{\left(m\right)}\in\mathcal{G}_{n,p}^{\left(m\right)}\left(\mathbb{L}\right):\; d_{c}\left(P^{\left(m\right)},Q^{\left(m\right)}\right)\le\delta\right\} .\]
 The volume of $B_{P^{\left(m\right)}}\left(\delta\right)$ as the
probability of an isotropically distributed $Q^{\left(m\right)}\in\mathcal{G}_{n,p}^{\left(m\right)}\left(\mathbb{L}\right)$
in this ball: \[
\mu\left(B_{P^{\left(m\right)}}\left(\delta\right)\right):=\Pr\left(Q^{\left(m\right)}\in B_{P^{\left(m\right)}}\left(\delta\right)\right).\]
 Since $\mu\left(B_{P^{\left(m\right)}}\left(\delta\right)\right)$
is independent of the choice of the center $P^{\left(m\right)}$,
we simply denote it by $\mu^{\left(m\right)}\left(\delta\right)$.
We have:

\begin{thm}
\label{thm:volume-CGM}When $\delta\le1$, \begin{equation}
\mu^{\left(m\right)}\left(\delta\right)=\frac{\Gamma^{m}\left(\frac{t}{2}+1\right)}{\Gamma\left(m\frac{t}{2}+1\right)}c_{n,p,p,\beta}^{m}\delta^{mt},\label{eq:volume-CGM}\end{equation}
where $c_{n,p,p,\beta}$ and $t$ are defined in Lemma \ref{lem:DRF-GM}. 
\end{thm}
\begin{proof}
Let us drop the subscript of $c_{n,p,p,\beta}$ during the proof.
In \cite{Dai_IT2008_Quantization_Grassmannian_manifold}, we proved
that for a single Grassmann manifold, $\mu^{\left(1\right)}\left(d_{c}^{2}\le x\right)=\mu^{\left(1\right)}\left(\sqrt{x}\right)=cx^{\frac{t}{2}}\left(1+O\left(x\right)\right)$
when $x\le1$, and it can be verified that \[
d\mu\left(d_{c}^{2}\le x\right)=\frac{t}{2}cx^{\frac{t}{2}-1}\left(1+O\left(x\right)\right)\cdot dx.\]
 By the definition of the volume, $d\mu^{\left(2\right)}\left(x\right)/dx$
is a convolution of $d\mu\left(x\right)/dx$ and $d\mu\left(x\right)/dx$.
So, \begin{align*}
\frac{d\mu^{\left(2\right)}\left(x\right)}{dx} & =\int_{0}^{x}\frac{t^{2}}{4}c^{2}\tau^{\frac{t}{2}-1}\left(x-\tau\right)^{\frac{t}{2}-1}\left(1+O\left(\tau\right)\right)\left(1+O\left(x-\tau\right)\right)d\tau\\
 & \overset{\left(a\right)}{=}\frac{t^{2}}{4}c^{2}x^{t-1}\int_{0}^{1}y^{\frac{t}{2}-1}\left(1-y\right)^{\frac{t}{2}-1}\left(1+O\left(xy\right)+O\left(x\left(1-y\right)\right)\right)dy\\
 & =\frac{t^{2}}{4}c^{2}x^{t-1}\frac{\Gamma\left(\frac{t}{2}\right)\Gamma\left(\frac{t}{2}\right)}{\Gamma\left(t\right)}\left(1+O\left(x\right)\right),\end{align*}
 where $\left(a\right)$ follows from the variable change $\tau=xy$.
A calculation produces \[
\mu^{\left(2\right)}\left(d_{c}^{2}\le x\right)=\frac{\Gamma\left(\frac{t}{2}+1\right)\Gamma\left(\frac{t}{2}+1\right)}{\Gamma\left(t+1\right)}c^{2}x^{t}\left(1+O\left(x\right)\right).\]
By mathematical induction, we reach (\ref{eq:volume-CGM}). Note that
$\delta\le1$ is required in every step. 
\end{proof}

Based on the volume formula, an upper bound on the distortion rate
function $D^{*}\left(K\right)$ on the composite Grassmann manifold
\[
D^{*}\left(K\right)\le\frac{2}{mt}\Gamma\left(\frac{2}{mt}\right)\frac{\Gamma^{\frac{2}{mt}}\left(m\frac{t}{2}+1\right)}{\Gamma^{\frac{2}{t}}\left(\frac{t}{2}+1\right)}c_{n,p,p,\beta}^{-\frac{2}{t}}2^{-\frac{2\log_{2}K}{mt}}\left(1+o\left(1\right)\right)\]
 is derived by calculating the average distortion of random codes
(see \cite{Dai_IT2008_Quantization_Grassmannian_manifold} for details).
Furthermore, by the sphere packing/covering argument (again see \cite{Dai_IT2008_Quantization_Grassmannian_manifold}
for details), the lower bound \[
\frac{mt}{mt+2}\frac{\Gamma^{\frac{2}{mt}}\left(mt+1\right)}{\Gamma^{\frac{2}{t}}\left(t+1\right)}c_{n,p,p,\beta}^{-\frac{2}{t}}2^{-\frac{2\log_{2}K}{mt}}\left(1+o\left(1\right)\right)\le D^{*}\left(K\right)\]
 is arrived at. Theorem \ref{thm:DRF-CGM} is proved.

\subsection{\label{sub:Proof-of-Theorem-n_bar}Proof of Theorem \ref{thm:n_bar}}

The key step is to prove that $\mathrm{E}_{\mathcal{B},\mathbf{H}}\left[\mathbf{V}_{k}^{\dagger}\mathbf{b}_{k}\mathbf{b}_{k}^{\dagger}\mathbf{V}_{k}\right]=\mathrm{diag}\left[\frac{\gamma}{s},\frac{s-\gamma}{s\left(L_{T}-1\right)},\cdots,\frac{s-\gamma}{s\left(L_{T}-1\right)}\right]$,
where $\mathbf{V}_{k}$ is from the singular value decomposition $\mathbf{H}_{\left(k\right)}=\mathbf{U}_{k}\mathbf{\Lambda}_{k}\bar{\mathbf{V}}_{k}^{\dagger}$.
Let $\mathbf{V}_{k}=\left[\mathbf{v}_{k}\bar{\mathbf{V}}_{k}\right]$
where $\bar{\mathbf{V}}_{k}\in\mathbb{C}^{L_{T}\times\left(L_{T}-1\right)}$
is composed of all the columns of $\mathbf{V}_{k}$ except $\mathbf{v}_{k}$.
Let $\mathbf{V}=\left[\mathbf{v}_{1}\cdots\mathbf{v}_{s}\right]$.
Recall our feedback function $\varphi\left(\mathbf{V}\right)$ in
(\ref{eq:quantization-fn-eigenchannels}) and definition of $\gamma$
in (\ref{eq:def-gamma}). Then the fact that \[
\mathrm{E}_{\mathcal{B},\mathbf{H}}\left[\mathbf{v}_{k}^{\dagger}\mathbf{b}_{k}\mathbf{b}_{k}^{\dagger}\mathbf{v}_{k}\right]=\frac{\gamma}{s}\]
is implied by the following lemma.

\begin{lemma}
\label{lem:correlation_v_k_b_k}Let $\mathbf{V}\in\mathcal{M}_{L_{T},1}^{\left(s\right)}$
be isotropically distributed and $\mathcal{B}\subset\mathcal{M}_{L_{T},1}^{\left(s\right)}$
be randomly generated from the isotropic distribution. Let $\mathbf{B}=\varphi\left(\mathbf{V}\right)$
where $\varphi\left(\cdot\right)$ is given in (\ref{eq:quantization-fn-eigenchannels})
and $\gamma$ is given by (\ref{eq:def-gamma}). Then \[
\mathrm{E}_{\mathcal{B},\mathbf{V}}\left[\mathbf{V}^{\dagger}\mathbf{B}\mathbf{B}^{\dagger}\mathbf{V}\right]=\frac{\gamma}{s}\mathbf{I}_{s}.\]

\end{lemma}
\begin{proof}
Let $\mathbf{Z}=\mathrm{E}_{\mathcal{B},\mathbf{V}}\left[\mathbf{V}^{\dagger}\mathbf{B}\mathbf{B}^{\dagger}\mathbf{V}\right]$.
For any $\theta\in\left[0,2\pi\right)$, let $\mathbf{A}_{k}=\mathrm{diag}\left[1,\cdots,1,e^{j\theta},1,\cdots,1\right]$
be obtained by replacing the $k^{\mathrm{th}}$ diagonal element of
$\mathbf{I}$ with $e^{j\theta}$. It can be verified that $\mathbf{V}\mathbf{A}_{k}\in\mathcal{M}_{L_{T},1}^{\left(s\right)}$
is isotropically distributed, and $\varphi\left(\mathbf{V}\mathbf{A}_{k}\right)=\varphi\left(\mathbf{V}\right)=\mathbf{B}$.
We have \begin{align*}
\mathbf{Z} & =\mathrm{E}_{\mathcal{B},\mathbf{V}\mathbf{A}_{k}}\left[\mathbf{A}_{k}^{\dagger}\mathbf{V}^{\dagger}\mathbf{B}\mathbf{B}^{\dagger}\mathbf{V}\mathbf{A}_{k}\right]\\
 & =\mathbf{A}_{k}^{\dagger}\mathrm{E}_{\mathcal{B},\mathbf{V}}\left[\mathbf{V}^{\dagger}\mathbf{B}\mathbf{B}^{\dagger}\mathbf{V}\right]\mathbf{A}_{k}\\
 & =\mathbf{A}_{k}^{\dagger}\mathbf{Z}\mathbf{A}_{k},\end{align*}
where the first equality is obtained by changing the variable from
$\mathbf{V}$ to $\mathbf{VA}_{k}$, and the second equality is obtained
by replacing the measure of $\mathbf{VA}_{k}$ with the measure of
$\mathbf{V}$. Then $\left(\mathbf{Z}\right)_{k,j}=e^{-j\theta}\left(\mathbf{Z}\right)_{k,j}$
for $j\ne k$, which is only possible if $\left(\mathbf{Z}\right)_{k,j}=0$.
Therefore, $\mathbf{Z}$ is a diagonal matrix. 

Now let $\mathbf{P}\in\mathbb{R}^{s\times s}$ be a permutation matrix
generated by permutating rows/columns of the identity matrix. Let
$\mathcal{B}\mathbf{P}=\left\{ \mathbf{BP}:\;\mathbf{B}\in\mathcal{B}\right\} $.
Then $\mathbf{VP}\in\mathcal{M}_{L_{T},1}^{\left(s\right)}$ and $\mathbf{BP}\in\mathcal{M}_{L_{T},1}^{\left(s\right)}$
are isotropically distributed. It can be verified that $\varphi_{\mathcal{B}\mathbf{P}}\left(\mathbf{VP}\right)=\mathbf{BP}=\varphi_{\mathcal{B}}\left(\mathbf{V}\right)\mathbf{P}$,
where the subscript $\varphi$ emphasizes the choice of codebook.
Then,\begin{align*}
\mathbf{Z} & =\mathrm{E}_{\mathcal{B}\mathbf{P},\mathbf{VP}}\left[\left(\mathbf{VP}\right)^{\dagger}\varphi_{\mathcal{B}\mathbf{P}}\left(\mathbf{VP}\right)\varphi_{\mathcal{B}\mathbf{P}}\left(\mathbf{VP}\right)^{\dagger}\left(\mathbf{VP}\right)\right]\\
 & =\mathbf{P}^{\dagger}\mathrm{E}_{\mathcal{B}\mathbf{P},\mathbf{V}}\left[\mathbf{V}^{\dagger}\varphi_{\mathcal{B}\mathbf{P}}\left(\mathbf{VP}\right)\varphi_{\mathcal{B}\mathbf{P}}\left(\mathbf{VP}\right)^{\dagger}\mathbf{V}\right]\mathbf{P}\\
 & =\mathbf{P}^{\dagger}\mathrm{E}_{\mathcal{B}\mathbf{P},\mathbf{V}}\left[\mathbf{V}^{\dagger}\varphi_{\mathcal{B}}\left(\mathbf{V}\right)\mathbf{P}\mathbf{P}^{\dagger}\varphi_{\mathcal{B}}\left(\mathbf{V}\right)^{\dagger}\mathbf{V}\right]\mathbf{P}\\
 & =\mathbf{P}^{\dagger}\mathrm{E}_{\mathcal{B},\mathbf{V}}\left[\mathbf{V}^{\dagger}\varphi_{\mathcal{B}}\left(\mathbf{V}\right)\varphi_{\mathcal{B}}\left(\mathbf{V}\right)^{\dagger}\mathbf{V}\right]\mathbf{P}\\
 & =\mathbf{P}^{\dagger}\mathbf{Z}\mathbf{P},\end{align*}
where the first equality is obtained by variables change, and the
second and fourth equality follows from measure replacement. It follows
that $\left(\mathbf{Z}\right)_{i,i}=\left(\mathbf{Z}\right)_{j,j}$
for $1\le i,j\le s$.

Finally, $\mathbf{Z}=\frac{\gamma}{s}\mathbf{I}$ follows from the
fact that $\mathrm{tr}\left(\mathbf{Z}\right)=\mathrm{E}\left[\mathrm{tr}\left(\mathbf{V}^{\dagger}\mathbf{B}\mathbf{B}^{\dagger}\mathbf{V}\right)\right]=\gamma$. 
\end{proof}

We evaluate \[
\mathrm{E}\left[\mathbf{V}_{k}^{\dagger}\mathbf{b}_{k}\mathbf{b}_{k}^{\dagger}\mathbf{V}_{k}\right]=\left[\begin{array}{cc}
\mathrm{E}\left[\mathbf{v}_{k}^{\dagger}\mathbf{b}_{k}\mathbf{b}_{k}^{\dagger}\mathbf{v}_{k}\right] & \mathrm{E}\left[\mathbf{v}_{k}^{\dagger}\mathbf{b}_{k}\mathbf{b}_{k}^{\dagger}\bar{\mathbf{V}}_{k}\right]\\
\mathrm{E}\left[\bar{\mathbf{V}}_{k}^{\dagger}\mathbf{b}_{k}\mathbf{b}_{k}^{\dagger}\mathbf{v}_{k}\right] & \mathrm{E}\left[\bar{\mathbf{V}}_{k}^{\dagger}\mathbf{b}_{k}\mathbf{b}_{k}^{\dagger}\bar{\mathbf{V}}_{k}\right]\end{array}\right].\]
For any unitary matrix $\mathbf{U}_{r}\in\mathbb{C}^{\left(L_{T}-1\right)\times\left(L_{T}-1\right)}$,
$\left[\mathbf{v}_{k},\bar{\mathbf{V}}_{k}\mathbf{U}\right]$ is also
isotropically distributed. Employ the method in the proof of Lemma
\ref{lem:correlation_v_k_b_k} to find that\[
\mathrm{E}\left[\mathbf{v}_{k}^{\dagger}\mathbf{b}_{k}\mathbf{b}_{k}^{\dagger}\bar{\mathbf{V}}_{k}\right]=\mathrm{E}\left[\mathbf{v}_{k}^{\dagger}\mathbf{b}_{k}\mathbf{b}_{k}^{\dagger}\bar{\mathbf{V}}_{k}\right]\mathbf{U},\]
\[
\mathrm{and}\quad\mathrm{E}\left[\bar{\mathbf{V}}_{k}^{\dagger}\mathbf{b}_{k}\mathbf{b}_{k}^{\dagger}\bar{\mathbf{V}}_{k}\right]=\mathbf{U}^{\dagger}\mathrm{E}\left[\bar{\mathbf{V}}_{k}^{\dagger}\mathbf{b}_{k}\mathbf{b}_{k}^{\dagger}\bar{\mathbf{V}}_{k}\right]\mathbf{U}.\]
Therefore, $\mathrm{E}\left[\mathbf{v}_{k}^{\dagger}\mathbf{b}_{k}\mathbf{b}_{k}^{\dagger}\bar{\mathbf{V}}_{k}\right]=\mathbf{o}^{\dagger}$
and $\mathrm{E}\left[\bar{\mathbf{V}}_{k}^{\dagger}\mathbf{b}_{k}\mathbf{b}_{k}^{\dagger}\bar{\mathbf{V}}_{k}\right]=c\mathbf{I}_{L_{T}-1}$
for some constant $c$. Note that $\mathrm{E}\left[\mathbf{v}_{k}^{\dagger}\mathbf{b}_{k}\mathbf{b}_{k}^{\dagger}\mathbf{v}_{k}\right]=\frac{\gamma}{s}$
and $\mathrm{E}\left[\mathrm{tr}\left(\mathbf{V}_{k}^{\dagger}\mathbf{b}_{k}\mathbf{b}_{k}^{\dagger}\mathbf{V}_{k}\right)\right]=1$.
Hence, $c=\frac{s-\gamma}{s\left(L_{T}-1\right)}$ and $\mathrm{E}\left[\mathbf{V}_{k}^{\dagger}\mathbf{b}_{k}\mathbf{b}_{k}^{\dagger}\mathbf{V}_{k}\right]=\mathrm{diag}\left[\frac{\gamma}{s},\frac{s-\gamma}{s\left(L_{T}-1\right)},\cdots,\frac{s-\gamma}{s\left(L_{T}-1\right)}\right]$.

Finally, \begin{align*}
\mathrm{E}_{\mathcal{B}}\left[\eta\right] & =\frac{1}{sL_{R}}\mathrm{E}_{\mathcal{B},\mathbf{H}}\left[\sum_{k=1}^{s}n_{k}^{2}\right]\\
 & =\frac{1}{sL_{R}}\sum_{k=1}^{s}\mathrm{E}_{\mathcal{B},\mathbf{H}}\left[\mathrm{tr}\left(\mathbf{H}_{\left(k\right)}\mathbf{b}_{k}\mathbf{b}_{k}^{\dagger}\mathbf{H}_{\left(k\right)}^{\dagger}\right)\right]\\
 & =\frac{1}{sL_{R}}\sum_{k=1}^{s}\mathrm{tr}\left(\mathrm{E}_{\mathcal{B},\mathbf{H}}\left[\mathbf{V}_{k}^{\dagger}\mathbf{b}_{k}\mathbf{b}_{k}^{\dagger}\mathbf{V}_{k}\right]\mathrm{E}_{\mathbf{H}}\left[\mathbf{\Lambda}_{k}^{\dagger}\mathbf{\Lambda}_{k}\right]\right)\\
 & =\frac{1}{sL_{R}}\sum_{k=1}^{s}\left(\frac{\gamma}{s}\zeta_{1}\mathrm{E}_{\mathbf{H}}\left[n_{\left(k\right)}^{2}\right]+\frac{s-\gamma}{s}\frac{\left(1-\zeta_{1}\right)\mathrm{E}_{\mathbf{H}}\left[n_{\left(k\right)}^{2}\right]}{L_{T}-1}\right)\\
 & =\frac{1}{L_{R}}\left(\frac{\gamma}{s}\zeta_{1}+\frac{s-\gamma}{s}\frac{1-\zeta_{1}}{L_{T}-1}\right)\bar{n}_{\left(\cdot\right)}^{2},\end{align*}
where the third line follows from the fact that $\mathbf{\Lambda}_{k}$
is independent of $\mathbf{V}_{k}$ and $\mathbf{b}_{k}$. 

\bibliographystyle{IEEEtran}
\bibliography{Bib/_Heath,Bib/_love,Bib/_Rao,Bib/_Tse,Bib/FeedbackMIMO_append,Bib/MIMO_basic,Bib/_Dai,Bib/RandomMatrix,Bib/_Verdu,Bib/Books,Bib/Multi_Access,Bib/Math,Bib/_Jindal,Bib/_BC_Feedback}

\end{document}